\documentclass[acmsmall]{acmart}

\setcopyright{cc}
\setcctype{by}
\acmDOI{10.1145/3808282}
\acmYear{2026}
\acmJournal{PACMPL}
\acmVolume{10}
\acmNumber{PLDI}
\acmArticle{204}
\acmMonth{6}
\acmSubmissionID{pldi26main-p183-p}
\received{2025-11-12}
\received[accepted]{2026-04-03}

\usepackage{paralist}
\usepackage{listings}
\usepackage{enumitem}
\usepackage{xcolor}
\usepackage{subcaption}
\usepackage{tcolorbox}
\usepackage{tikz}
\usetikzlibrary{calc,positioning}
\usepackage{wrapfig}
\usepackage{algorithm}
\usepackage[noend]{algpseudocode}
\usepackage{thmtools}
\usepackage{xspace}

\newcommand{\SA}{\texttt{SA}\xspace}
\newcommand{\LCS}{\texttt{LCS}\xspace}
\newcommand{\FMBR}{\texttt{FMBR}\xspace}
\newcommand{\XThC}{\texttt{X3C}\xspace}

\newcommand{\pds}{\mathcal{S}}
\newcommand{\semiring}{\mathcal{D}}
\newcommand{\semiringSet}{D}
\newcommand{\pdsstates}{S}
\newcommand{\pdsstk}{\Gamma}
\newcommand{\pdstrans}{\Delta}
\newcommand{\pdsmove}{\hookrightarrow}
\newcommand{\zero}{\overline{0}}
\newcommand{\one}{\overline{1}}
\newcommand{\add}{\oplus}
\newcommand{\mult}{\otimes}
\newcommand{\wt}{\text{wt}}
\newcommand{\paths}{\text{Paths}}
\newcommand{\reach}{\text{Reach}}

\newcommand{\PTIME}{\textit{\textbf{P}}\xspace}
\newcommand{\NP}{\textit{\textbf{NP}}\xspace}
\newcommand{\XP}{\textit{\textbf{XP}}\xspace}
\newcommand{\FPT}{\textit{\textbf{FPT}}\xspace}

\newcommand{\func}{\texttt{F}}
\newcommand{\param}{\eta}
\newcommand{\bigO}{\mathcal{O}}

\newcommand{\br}{~\textbf{|}~}
\newcommand{\seq}{\cdot}

\newcommand{\progs}{\text{Prog}}
\newcommand{\reOrd}{\pi}
\newcommand{\linearize}{\text{Linear}}
\newcommand{\brFa}{\text{BrFactor}}
\newcommand{\brID}{\text{BrID}}
\newcommand{\neDe}{\text{Dep}}

\newcommand{\Ebren}{E^\text{en}}
\newcommand{\Ebrex}{E^\text{ex}}
\newcommand{\Eseq}{E^{\text{seq}}}
\newcommand{\brtext}{\textbf{br}}
\newcommand{\lab}{\text{Lab}}
\newcommand{\BG}{\text{BrGraph}}

\newcommand{\score}{\text{Score}}
\newcommand{\reg}{\text{Region}}
\newcommand{\emptysgn}{[]}

\algrenewcommand\algorithmicrequire{\textbf{Input:}}
\algrenewcommand\algorithmicensure{\textbf{Output:}}

\definecolor{azure(web)(azuremist)}{rgb}{0.94, 1.0, 1.0}
\definecolor{ashgrey}{rgb}{0.7, 0.75, 0.71}
\definecolor{almond}{rgb}{0.94, 0.87, 0.8}
\definecolor{bleudefrance}{rgb}{0.19, 0.55, 0.91}
\definecolor{blue(pigment)}{rgb}{0.2, 0.2, 0.6}
\definecolor{bostonuniversityred}{rgb}{0.8, 0.0, 0.0}
\definecolor{burgundy}{rgb}{0.5, 0.0, 0.13}
\definecolor{coolblack}{rgb}{0.0, 0.18, 0.39}
\definecolor{britishracinggreen}{rgb}{0.0, 0.26, 0.15}
\definecolor{airforceblue}{rgb}{0.36, 0.54, 0.66}
\definecolor{azure(colorwheel)}{rgb}{0.0, 0.5, 1.0}
\definecolor{ao(english)}{rgb}{0.0, 0.5, 0.0}
\definecolor{cbblue}{RGB}{0,114,178}   
\definecolor{cbred}{RGB}{214,94,0}
\definecolor{black}{RGB}{0,0,0}
\definecolor{gray}{RGB}{100,100,100}

\lstdefinestyle{pythonstyle}{
  language=Python,
  basicstyle=\ttfamily\normalsize\color{black},
  keywordstyle=\color{britishracinggreen}\bfseries,
  stringstyle=\color{cbblue},
  commentstyle=\color{gray},
  showstringspaces=false, 
  breaklines=true,
  tabsize=2,
  numbers=left,
  numberstyle=\scriptsize\color{black}\ttfamily, 
  numbersep=3pt,                      
  literate={+}{{\bfseries\textcolor{burgundy}{+}}}1
           {-}{{\bfseries\textcolor{burgundy}{-}}}1
           {*}{{\bfseries\textcolor{burgundy}{*}}}1
           {/}{{\bfseries\textcolor{burgundy}{/}}}1
           {=}{{\bfseries\textcolor{burgundy}{=}}}1
           {+=}{{\bfseries\textcolor{burgundy}{+=}}}1
           {-=}{{\bfseries\textcolor{burgundy}{-=}}}1
           {*=}{{\bfseries\textcolor{burgundy}{*=}}}1
           {/=}{{\bfseries\textcolor{burgundy}{/=}}}1,
  mathescape=true
}

\newenvironment{myitemize}{
    \begin{itemize}[noitemsep,topsep=0pt,parsep=0pt,partopsep=0pt,leftmargin=*]
}{
    \end{itemize}
}

\newenvironment{myenumerate}{
    \begin{enumerate}[noitemsep,topsep=0pt,parsep=0pt,partopsep=0pt,leftmargin=*]
}{
    \end{enumerate}
}

\begin{document}

\title{Parameterized Algorithms and Complexity for Function Merging with Branch Reordering}

\author{Amir K. Goharshady}
\orcid{0000-0003-1702-6584}
\affiliation{
  \institution{University of Oxford}
  \country{United Kingdom}
}
\email{amir.goharshady@cs.ox.ac.uk}

\author{Kerim Kochekov}
\orcid{0009-0000-8255-1831}
\affiliation{
  \institution{Hong Kong University of Science and Technology}
  \country{Hong Kong}
}
\email{kkochekov@connect.ust.hk}

\author{Tian Shu}
\orcid{0009-0004-5649-9469}
\affiliation{
  \institution{Hong Kong University of Science and Technology}
  \country{Hong Kong}
}
\email{tshu@connect.ust.hk}

\author{Ahmed Khaled Zaher}
\orcid{0000-0001-5894-7991}
\affiliation{
  \institution{Hong Kong University of Science and Technology}
  \country{Hong Kong}
}
\email{akazaher@connect.ust.hk}
\authornote{Corresponding Author}

\renewcommand{\shortauthors}{Goharshady et al.}

\begin{abstract}
Binary size reduction is an increasingly important optimization objective for compilers, especially in the context of mobile applications and resource-constrained embedded devices. In such domains, binary size often takes precedence over compilation time. One emerging technique that has been shown effective is \emph{function merging}, where multiple similar functions are merged into one, thereby eliminating redundancy. The state-of-the-art approach to perform the merging, due to Rocha et al. [CGO 2019, PLDI 2020], is based on \emph{sequence alignment}, where functions are viewed as \emph{linear} sequences of instructions that are then matched in a way maximizing their alignment.

In this paper, we consider a significantly generalized formulation of the problem by allowing \emph{reordering of branches} within each function, subsequently allowing for more flexible matching and better merging. We show that this makes the problem \textit{\textbf{NP}}-hard, and thus we study it through the lens of \emph{parameterized algorithms and complexity}, where we identify certain parameters of the input that govern its complexity. We look at two natural parameters: the \emph{branching factor} and \emph{nesting depth} of input functions.

Concretely, our input consists of two functions $\texttt{F}_1, \texttt{F}_2,$ where each $\texttt{F}_i$ has size $n_i,$ branching factor $b_i,$ and nesting depth $d_i.$ Our task is to reorder the branches of $\texttt{F}_1$ and $\texttt{F}_2$ in a way that yields linearizations achieving the maximum sequence alignment. Let $n=\max(n_1, n_2),$ and define $b, d$ similarly. Our results are as follows: 
\begin{itemize}[noitemsep,topsep=0pt,parsep=0pt,partopsep=0pt,leftmargin=*]
  \item A simple algorithm running in time $2^{\mathcal{O}(bd)}  n^2,$ establishing that the problem is fixed-parameter tractable~(\textit{\textbf{FPT}}) with respect to all four parameters $b_1,d_1, b_2, d_2.$
  \item An algorithm running in time $2^{\mathcal{O}(bd_2)}  n^7,$ showing that even when one of the functions has an unbounded nesting depth, the problem remains in \textit{\textbf{FPT}}.
  \item A hardness result showing that the problem is \textit{\textbf{NP}}-hard even when constrained to constant $d_1, b_2, d_2.$
\end{itemize}
To the best of our knowledge, this is the first systematic study of function merging with branch reordering from an algorithmic or complexity-theoretic perspective.
\end{abstract}



\maketitle


\section{Introduction}\label{sec:introduction}

\paragraph{Binary Size Optimization} Binary size reduction is an increasingly important optimization objective for compilers, and there has been an influx of works in this area from researchers in academia and industry across the compilers and embedded systems communities~\cite{postiz,profguided,rollbin,linkersz,looproll,reordering,mobile,f3m,hybf,newfm,inlinemobile}. Two key applications that drive this interest are mobile applications, e.g.,~iOS and Android apps, as well as resource-constrained embedded devices, e.g.,~Internet-of-Things (IoT) devices. As new features and updates are continuously added to mobile applications, their binary size is often compromised~\cite{mobile}. At the same time, software vendors usually impose a maximum limit on the size of apps that can be downloaded, and users tend to avoid downloading large apps~\cite{mobile,google}. In embedded devices, it is typical to have minimal computational capabilities available in order to reduce production cost~\cite{iot}. In such devices, large code size can occupy a significant portion of storage and lead to increased power consumption~\cite{rollbin}. These reasons necessitate effective size minimization techniques, even at the cost of longer compilation time.

\paragraph{Function Merging} One emerging technique that has been shown effective for binary size reduction is \emph{function merging}~\cite{funcsim,seqalign,ssafuncmerg,hyfm,f3m,multimerge,newfm}. The high-level idea is to repeatedly run the following process: Identify a pair of functions $\func_1,\func_2$ that are highly similar, under some notion of similarity, and introduce a new function $\func_\star$ that captures their combined behaviour. The parts common to both $\func_1$ and $\func_2$ will be included only once in $\func_\star,$ while the differing parts will be captured through an additional parameter given to $\func_\star,$ which determines which $\func_i$ to simulate. Finally, the body of each $\func_i$ is replaced with a single call to $\func_\star$ with the appropriate argument. The more similar $\func_1$ and $\func_2$ are, the more success merging can have in reducing code size. Thus, the usefulness of this optimization relies on the presence of redundancy in the code. Such redundancy can arise, for instance, due to using templates in a high-level language, which is projected into multiple highly-similar copies in the low-level code, to which the optimization is applied. Moreover, interprocedural compiler optimizations can create more redundancy in code which can then be eliminated through merging. See~\cite[Section 1]{funcsim} for further discussion.

\paragraph{Remark} We focus on the case of merging 2 functions for simplicity of exposition, but all of our results can be easily generalized to merging $k\geq 3$ functions.

\paragraph{Running Example} We will use the Python code in Figure~\ref{fig:F1F2seqalign} as a running example throughout the paper. It has two functions $\func_1$ (left) and $\func_2$ (right) that we wish to merge. One way to perform the merging is shown in Figure~\ref{fig:F1F2seqalignmerged} (left), where we merge the shared code portions in $\func_1,\func_2$ highlighted in green. Note that such instructions in $\func_1$ and $\func_2$ are not only equivalent, but also appear in the same order, i.e.,~they are \emph{aligned} across $\func_1$ and $\func_2.$ Other code that only appears in $\func_1$ or $\func_2$ is highlighted in red and blue, respectively, and no merging is performed on them. In the merged function $\func_\star,$ we introduce a new parameter \texttt{b}. Shared code is executed unconditionally, and code appearing only in $\func_1$ or $\func_2$ is executed depending on the value of \texttt{b}. Finally, the bodies of $\func_1$ and $\func_2$ are replaced with calls to the merged function, as shown in Figure~\ref{fig:F1F2seqalignmerged} (right). Suppose for simplicity that code size is measured by the number of lines of code. Then, before merging, the total size of $\func_1$ and $\func_2$ is $29$ lines, while after merging, the total size is $25$ lines, achieving a size reduction of $4$ lines.

\begin{figure}[h]
  \centering
  \hspace{0.03\linewidth}
  \begin{subfigure}{0.46\linewidth}
    \centering
    \scalebox{00.85}{
    \begin{tikzpicture}
      \node[anchor=north west, inner sep=0] (code) at (0, 0) {
\begin{lstlisting}[firstline=1, lastline=15, style=pythonstyle]
def F$_1$(x):
  x += 1
  y = x * 75
  z = -2 * x
  print(y)
  if x == 0:
    print(x)
    for i in range(3):
      print(i)
      print(i**2)
  elif x == 1:
    print(z)
    print(x * z)
  print(x**2)
  print(z**3)
\end{lstlisting}\par
      };

      \pgfmathsetmacro{\lineLen}{0.458}
      \pgfmathsetmacro{\shift}{0.14}
      \pgfmathsetmacro{\decr}{0.07}

      \pgfmathsetmacro{\X}{1}
      \pgfmathsetmacro{\len}{1}

      \draw[ao(english)!100, thick, rounded corners, fill=green!10, fill opacity=0.1] ($(code.north west) + (0,-\shift-\X * \lineLen)$) rectangle ($(code.north east) + (0, -\shift)+ (0, -\X * \lineLen - \len * \lineLen + \decr)$);

      \pgfmathsetmacro{\X}{2}
      \pgfmathsetmacro{\len}{1}

      \draw[bostonuniversityred!100, thick, rounded corners, fill=red!10, fill opacity=0.1] ($(code.north west) + (0,-\shift-\X * \lineLen)$) rectangle ($(code.north east) + (0, -\shift)+ (0, -\X * \lineLen - \len * \lineLen + \decr)$);

      \pgfmathsetmacro{\X}{3}
      \pgfmathsetmacro{\len}{1}

      \draw[ao(english)!100, thick, rounded corners, fill=green!10, fill opacity=0.1] ($(code.north west) + (0,-\shift-\X * \lineLen)$) rectangle ($(code.north east) + (0, -\shift)+ (0, -\X * \lineLen - \len * \lineLen + \decr)$);

      \pgfmathsetmacro{\X}{4}
      \pgfmathsetmacro{\len}{1}

      \draw[bostonuniversityred!100, thick, rounded corners, fill=red!10, fill opacity=0.1] ($(code.north west) + (0,-\shift-\X * \lineLen)$) rectangle ($(code.north east) + (0, -\shift)+ (0, -\X * \lineLen - \len * \lineLen + \decr)$);

      \pgfmathsetmacro{\X}{5}
      \pgfmathsetmacro{\len}{5}

      \draw[ao(english)!100, thick, rounded corners, fill=green!10, fill opacity=0.1] ($(code.north west) + (0,-\shift-\X * \lineLen)$) rectangle ($(code.north east) + (0, -\shift)+ (0, -\X * \lineLen - \len * \lineLen + \decr)$);

      \pgfmathsetmacro{\X}{10}
      \pgfmathsetmacro{\len}{3}

      \draw[bostonuniversityred!100, thick, rounded corners, fill=red!10, fill opacity=0.1] ($(code.north west) + (0,-\shift-\X * \lineLen)$) rectangle ($(code.north east) + (0, -\shift)+ (0, -\X * \lineLen - \len * \lineLen + \decr)$);

      \pgfmathsetmacro{\X}{13}
      \pgfmathsetmacro{\len}{2}

      \draw[ao(english)!100, thick, rounded corners, fill=green!10, fill opacity=0.1] ($(code.north west) + (0,-\shift-\X * \lineLen)$) rectangle ($(code.north east) + (0, -\shift)+ (0, -\X * \lineLen - \len * \lineLen + \decr)$);
    \end{tikzpicture}
    }
  \end{subfigure}
  \hfill
  \begin{subfigure}{0.48\linewidth}
    \centering
    \scalebox{00.85}{
    \begin{tikzpicture}
      \node[anchor=north west, inner sep=0] (code) at (0, 0) {
\begin{lstlisting}[firstline=18, lastline=31, firstnumber=16, style=pythonstyle]
def F$_2$(x):
  x += 1
  z = -2 * x
  print(x + z)
  if x == 1:
    print(z)
    print(x * z)
  elif x == 0:
    print(x)
    for i in range(3):
      print(i)
      print(i**2)
  print(x**2)
  print(z**3)
\end{lstlisting}\par
      };

      \pgfmathsetmacro{\lineLen}{0.458}
      \pgfmathsetmacro{\shift}{0.14}
      \pgfmathsetmacro{\decr}{0.07}

      \pgfmathsetmacro{\X}{1}
      \pgfmathsetmacro{\len}{1}

      \draw[ao(english)!100, thick, rounded corners, fill=green!10, fill opacity=0.1] ($(code.north west) + (0,-\shift-\X * \lineLen)$) rectangle ($(code.north east) + (0, -\shift)+ (0, -\X * \lineLen - \len * \lineLen + \decr)$);

      \pgfmathsetmacro{\X}{2}
      \pgfmathsetmacro{\len}{1}

      \draw[ao(english)!100, thick, rounded corners, fill=green!10, fill opacity=0.1] ($(code.north west) + (0,-\shift-\X * \lineLen)$) rectangle ($(code.north east) + (0, -\shift)+ (0, -\X * \lineLen - \len * \lineLen + \decr)$);

      \pgfmathsetmacro{\X}{3}
      \pgfmathsetmacro{\len}{1}

      \draw[blue(pigment)!100, thick, rounded corners, fill=blue!10, fill opacity=0.1] ($(code.north west) + (0,-\shift-\X * \lineLen)$) rectangle ($(code.north east) + (0, -\shift)+ (0, -\X * \lineLen - \len * \lineLen + \decr)$);

      \pgfmathsetmacro{\X}{4}
      \pgfmathsetmacro{\len}{3}

      \draw[blue(pigment)!100, thick, rounded corners, fill=blue!10, fill opacity=0.1] ($(code.north west) + (0,-\shift-\X * \lineLen)$) rectangle ($(code.north east) + (0, -\shift)+ (0, -\X * \lineLen - \len * \lineLen + \decr)$);

      \pgfmathsetmacro{\X}{7}
      \pgfmathsetmacro{\len}{5}

      \draw[ao(english)!100, thick, rounded corners, fill=green!10, fill opacity=0.1] ($(code.north west) + (0,-\shift-\X * \lineLen)$) rectangle ($(code.north east) + (0, -\shift)+ (0, -\X * \lineLen - \len * \lineLen + \decr)$);

      \pgfmathsetmacro{\X}{12}
      \pgfmathsetmacro{\len}{2}

      \draw[ao(english)!100, thick, rounded corners, fill=green!10, fill opacity=0.1] ($(code.north west) + (0,-\shift-\X * \lineLen)$) rectangle ($(code.north east) + (0, -\shift)+ (0, -\X * \lineLen - \len * \lineLen + \decr)$);
    \end{tikzpicture}
    }
    \vspace{.7em}
  \end{subfigure}
  \caption{Two functions $\func_1$ (left) and $\func_2$ (right) to be merged. The shared code is highlighted in green, whereas code only in $\func_1$ or $\func_2$ is highlighted in red and blue, respectively.}
  \label{fig:F1F2seqalign}
\end{figure}

\begin{figure}[h]
\vspace{-1.5em}
  \centering
  \hspace{0.03\linewidth}
  \begin{subfigure}{0.46\linewidth}
    \centering
    \scalebox{00.85}{
    \begin{tikzpicture}
      \node[anchor=north west, inner sep=0] (code) at (0, 0) {
\begin{lstlisting}[firstline=34, lastline=55, firstnumber=1, style=pythonstyle]
def F$_\star$(x, b):
  x += 1
  if b == 1:
    y = x * 75
  z = -2 * x
  if b == 1:
    print(y)
  else:
    print(x + z)
  if b == 2 and x == 1:
    print(z)
    print(x * z)
  elif x == 0:
    print(x)
    for i in range(3):
      print(i)
      print(i**2)
  elif b == 1 and x == 1:
    print(z)
    print(x * z)
  print(x**2)
  print(z**3)
\end{lstlisting}\par
      };
      \pgfmathsetmacro{\lineLen}{0.458}
      \pgfmathsetmacro{\shift}{0.14}
      \pgfmathsetmacro{\decr}{0.07}

      \pgfmathsetmacro{\X}{1}
      \pgfmathsetmacro{\len}{1}

      \draw[ao(english)!100, thick, rounded corners, fill=green!10, fill opacity=0.1] ($(code.north west) + (0,-\shift-\X * \lineLen)$) rectangle ($(code.north east) + (0, -\shift)+ (0, -\X * \lineLen - \len * \lineLen + \decr)$);

      \pgfmathsetmacro{\X}{2}
      \pgfmathsetmacro{\len}{2}

      \draw[bostonuniversityred!100, thick, rounded corners, fill=red!10, fill opacity=0.1] ($(code.north west) + (0,-\shift-\X * \lineLen)$) rectangle ($(code.north east) + (0, -\shift)+ (0, -\X * \lineLen - \len * \lineLen + \decr)$);

      \pgfmathsetmacro{\X}{4}
      \pgfmathsetmacro{\len}{1}

      \draw[ao(english)!100, thick, rounded corners, fill=green!10, fill opacity=0.1] ($(code.north west) + (0,-\shift-\X * \lineLen)$) rectangle ($(code.north east) + (0, -\shift)+ (0, -\X * \lineLen - \len * \lineLen + \decr)$);

      \pgfmathsetmacro{\X}{5}
      \pgfmathsetmacro{\len}{2}

      \draw[bostonuniversityred!100, thick, rounded corners, fill=red!10, fill opacity=0.1] ($(code.north west) + (0,-\shift-\X * \lineLen)$) rectangle ($(code.north east) + (0, -\shift)+ (0, -\X * \lineLen - \len * \lineLen + \decr)$);

      \pgfmathsetmacro{\X}{7}
      \pgfmathsetmacro{\len}{2}
      
      \draw[blue(pigment)!100, thick, rounded corners, fill=blue!10, fill opacity=0.1] ($(code.north west) + (0,-\shift-\X * \lineLen)$) rectangle ($(code.north east) + (0, -\shift)+ (0, -\X * \lineLen - \len * \lineLen + \decr)$);

      \pgfmathsetmacro{\X}{9}
      \pgfmathsetmacro{\len}{3}
      
      \draw[blue(pigment)!100, thick, rounded corners, fill=blue!10, fill opacity=0.1] ($(code.north west) + (0,-\shift-\X * \lineLen)$) rectangle ($(code.north east) + (0, -\shift)+ (0, -\X * \lineLen - \len * \lineLen + \decr)$);

      \pgfmathsetmacro{\X}{12}
      \pgfmathsetmacro{\len}{5}

      \draw[ao(english)!100, thick, rounded corners, fill=green!10, fill opacity=0.1] ($(code.north west) + (0,-\shift-\X * \lineLen)$) rectangle ($(code.north east) + (0, -\shift)+ (0, -\X * \lineLen - \len * \lineLen + \decr)$);

      \pgfmathsetmacro{\X}{17}
      \pgfmathsetmacro{\len}{3}

      \draw[bostonuniversityred!100, thick, rounded corners, fill=red!10, fill opacity=0.1] ($(code.north west) + (0,-\shift-\X * \lineLen)$) rectangle ($(code.north east) + (0, -\shift)+ (0, -\X * \lineLen - \len * \lineLen + \decr)$);

      \pgfmathsetmacro{\X}{20}
      \pgfmathsetmacro{\len}{2}

      \draw[ao(english)!100, thick, rounded corners, fill=green!10, fill opacity=0.1] ($(code.north west) + (0,-\shift-\X * \lineLen)$) rectangle ($(code.north east) + (0, -\shift)+ (0, -\X * \lineLen - \len * \lineLen + \decr)$);

    \end{tikzpicture}
    }
  \end{subfigure}
  \hfill
  \begin{subfigure}{0.48\linewidth}
    \centering

    \scalebox{00.85}{
    \begin{tikzpicture}
      \node[anchor=north west, inner sep=0] (code) at (0, 0) {
\begin{lstlisting}[firstline=77, lastline=78, firstnumber=23, style=pythonstyle]
def F$_1$(x):
  F$_\star$(x, 1)
\end{lstlisting}\par
      };
    \end{tikzpicture}
    }

    ~

    \scalebox{00.85}{
    \begin{tikzpicture}
      \node[anchor=north west, inner sep=0] (code) at (0, 0) {
\begin{lstlisting}[firstline=79, lastline=80, firstnumber=25, style=pythonstyle]
def F$_2$(x):
  F$_\star$(x, 2)
\end{lstlisting}\par
      };
    \end{tikzpicture}
    }
    \vspace{10em}
  \end{subfigure}
  \caption{Merged function $\func_\star$ (left) and the updated bodies of $\func_1$ and $\func_2$ after the merge (right).}
  \label{fig:F1F2seqalignmerged}
\end{figure}

\paragraph{Pair Selection vs. Merging} We distinguish between the following two tasks in function merging: 
\begin{compactenum}
  \item[\textbf{Pair Selection:}] From $n$ functions $\func_1,\dots,\func_n,$ optimally \emph{select then merge} a pair of functions $\func_{i_1}, \func_{i_2}.$
  \item[\textbf{Merging}:] Optimally merge a \emph{given} pair of functions $\func_1, \func_2.$
\end{compactenum}
As the pair selection problem involves not only merging a pair of functions optimally, but also finding the best pair to merge, it clearly subsumes the merging problem. At the same time, pair selection is too difficult, as it seems to require $\Theta(n^2)$ invocations of the merging routine. For this reason, all related works~\cite{seqalign,ssafuncmerg,f3m,funcsim} take a heuristic approach to pair selection by first identifying a \emph{seemingly promising} pair of functions, and then delegating the extraction of as much profit from that pair to the merging routine. In this work, we only focus on the merging task, and whenever we speak of function merging, we are referring to optimally merging two functions that are readily given as part of the input.

\paragraph{Sequence Alignment} \emph{Sequence alignment (SA)} is a classical problem that has its roots in bioinformatics, and it generalizes other standard problems such as \emph{edit distance} and \emph{longest common subsequence}~\cite{bioinfintro}. Informally, the input consists of two strings/sequences $S_1, S_2$ of $n_1, n_2$ characters, respectively, and our goal is to insert or delete characters from them in order to maximize their character-wise alignment. This optimal alignment has a value denoted by $\SA(S_1, S_2).$ Formal definitions are delegated to~Section~\ref{sec:formulation}, but we note that $\SA(S_1, S_2)$ and its corresponding alignment can be computed in time $\bigO\left(n_1 n_2\right)$ via a standard dynamic programming algorithm~\cite{bioinfintro,clrs}.

\paragraph{Function Merging via Sequence Alignment~\cite{seqalign,ssafuncmerg}} The state-of-the-art approach to perform function merging is based on sequence alignment~\cite{seqalign,ssafuncmerg}. Their insight is to view each instruction as a character, and each function $\func_i$ ($i\in\{1, 2\}$) as a string $S_i$ of such characters. Then, they run an SA algorithm to compute $\SA(S_1, S_2)$ and find the best alignment, which then guides the merging of the two functions. This is precisely the approach used to perform the merging in Figures~\ref{fig:F1F2seqalign} and~\ref{fig:F1F2seqalignmerged}: The SA algorithm identifies the aligned regions in $\func_1$ and $\func_2$ highlighted in green, and moreover guarantees that this is the best possible alignment that, say, maximizes the number of aligned instructions. Given this alignment, a merging procedure merges the aligned instructions (green) together, while the unaligned instructions (red and blue) are kept separate, as explained above.

\paragraph{Control-flow Graphs} A major limitation of the SA-based approach is that it treats the code of a function as a fixed \emph{linear} string of instructions, whereas it can be represented more generally as a \emph{control-flow graph (CFG)}~\cite{cfgs}. A CFG of a function $\func$ is a directed graph $G_\func=(V_\func, E_\func),$ where each node $v\in V_\func$ represents a basic block, which models sequential statements, and an edge $(u, v)\in E_\func$ indicates that the control can flow from basic block $u$ to basic block $v,$ which models branching constructs such as \texttt{if-else} statements and loops.\footnote{We treat a loop as a branching construct with one branch, the loop's body.} Given a CFG $G_\func$ and a \emph{traversal order} of $V_\func,$ a \emph{linearization} procedure produces a function $\func'$ respecting this order~\cite{seqalign}. $\func'$ is equivalent to $\func,$ but may or may not syntactically coincide with it. Throughout this paper, the notion of \emph{equivalence} is used informally to indicate that the different syntactic representations in question (e.g., $\func$ and $\func'$) achieve the same function.

As an example, Figure~\ref{fig:cfgandreorder} (left) shows the CFG $G_{\func_1}$ of the function $\func_1$ in Figure~\ref{fig:F1F2seqalign}. At the end of the basic block at $v_1,$ we are free to choose, in the linearized program, whether to put the branch $\{v_2, v_3\}$ (corresponding to the $\texttt{x == 0}$ case) before or after the branch $\{v_4\}$ (corresponding to the $\texttt{x == 1}$ case). The former choice yields $\func_1,$ whereas the latter yields the function $\func_1'$ shown in Figure~\ref{fig:cfgandreorder}~(right), where we swap the two branches.

\begin{figure}[h]
  \centering
  \begin{subfigure}{0.60\linewidth}
    
    \scalebox{00.85}{
    \begin{tikzpicture}[
        node distance=2cm and 2.5cm,
        every node/.style={draw, rounded corners, align=left, font=\ttfamily\footnotesize},
        edge from parent/.style={draw, thick, ->, >=stealth},
        edge label/.style={midway, fill=azure(web)(azuremist), font=\footnotesize\itshape}
      ]

      \node (v1) [fill=almond!10] {
        $v_1$:\\
\begin{lstlisting}[firstline=86, lastline=89, style=pythonstyle,numbers=none]
x += 1
y = x * 75
z = -2 * x
print(y)
\end{lstlisting}\par
      };
      \node (v2) [below left=of v1,xshift=2cm,yshift=1.35cm, fill=almond!10] {
        $v_2$:\\
\begin{lstlisting}[firstline=121, lastline=122, style=pythonstyle,numbers=none]
print(x)
i = 0
\end{lstlisting}\par
      };
      \node (v3) [below=of v2, xshift=0.4cm, yshift=1.1cm,fill=almond!10] {
        $v_3$:\\
\begin{lstlisting}[firstline=126, lastline=128, style=pythonstyle,numbers=none]
print(i)
print(i**2)
i += 1
\end{lstlisting}\par
      };

      \node (v4) [below right=of v1, xshift=-2.9cm, yshift=-1.2cm, fill=almond!10] {
        $v_4$:\\
\begin{lstlisting}[firstline=96, lastline=97, style=pythonstyle,numbers=none]
print(z)
print(x * z)
\end{lstlisting}\par
      };
      \node (v5) [below=of v3, xshift=2.5cm, yshift=1cm, fill=almond!10] {
        $v_5$:\\
\begin{lstlisting}[firstline=98, lastline=99, style=pythonstyle,numbers=none]
print(x**2)
print(z**3)
\end{lstlisting}\par
        };

      \draw[->] (v1) to[out=200,in=80,looseness=0.9] node[edge label] {\texttt{x == 0}} (v2);
      \draw[->] (v1) to[out=-18,in=82,looseness=0.75] node[edge label] {\texttt{x == 1}} (v4);
      \draw[->] (v2) to[out=-90,in=110,looseness=0.5] (v3);
      \draw[->, loop above, looseness=4, out=120, in=160] (v3) to node[edge label,yshift=-5pt] {\texttt{i < 3}} (v3) ;
      \draw[->] (v3) to[out=-90,in=160,looseness=0.8] node[edge label] {\texttt{i >= 3}} (v5);
      \draw[->] (v4) to[out=-95,in=38,looseness=0.7] (v5);
      \draw[->] (v1) to[out=-75, in=120] node[edge label, pos=0.3] {\texttt{x != 0 and x != 1}} (v5);

    \end{tikzpicture}
    }
  \end{subfigure}
  \hfill
  \begin{subfigure}{0.36\linewidth}
    \centering

    \scalebox{00.85}{
    \begin{tikzpicture}
      \node[anchor=north west, inner sep=0] (code) at (0, 0) {
\begin{lstlisting}[firstline=103, lastline=117, style=pythonstyle]
def F$_1'$(x):
  x += 1
  y = x * 75
  z = -2 * x
  print(y)
  if x == 1:
    print(z)
    print(x * z)
  elif x == 0:
    print(x)
    for i in range(3):
      print(i)
      print(i**2)
  print(x**2)
  print(z**3)
\end{lstlisting}\par
      };
      \pgfmathsetmacro{\lineLen}{0.458}
      \pgfmathsetmacro{\shift}{0.14}
      \pgfmathsetmacro{\decr}{0.07}

      \pgfmathsetmacro{\X}{1}
      \pgfmathsetmacro{\len}{1}

      \draw[ao(english)!100, thick, rounded corners, fill=green!10, fill opacity=0.1] ($(code.north west) + (0,-\shift-\X * \lineLen)$) rectangle ($(code.north east) + (0, -\shift)+ (0, -\X * \lineLen - \len * \lineLen + \decr)$);

      \pgfmathsetmacro{\X}{2}
      \pgfmathsetmacro{\len}{1}

      \draw[bostonuniversityred!100, thick, rounded corners, fill=red!10, fill opacity=0.1] ($(code.north west) + (0,-\shift-\X * \lineLen)$) rectangle ($(code.north east) + (0, -\shift)+ (0, -\X * \lineLen - \len * \lineLen + \decr)$);

      \pgfmathsetmacro{\X}{3}
      \pgfmathsetmacro{\len}{1}

      \draw[ao(english)!100, thick, rounded corners, fill=green!10, fill opacity=0.1] ($(code.north west) + (0,-\shift-\X * \lineLen)$) rectangle ($(code.north east) + (0, -\shift)+ (0, -\X * \lineLen - \len * \lineLen + \decr)$);

      \pgfmathsetmacro{\X}{4}
      \pgfmathsetmacro{\len}{1}

      \draw[bostonuniversityred!100, thick, rounded corners, fill=red!10, fill opacity=0.1] ($(code.north west) + (0,-\shift-\X * \lineLen)$) rectangle ($(code.north east) + (0, -\shift)+ (0, -\X * \lineLen - \len * \lineLen + \decr)$);

      \pgfmathsetmacro{\X}{5}
      \pgfmathsetmacro{\len}{8}

      \draw[ao(english)!100, thick, rounded corners, fill=green!10, fill opacity=0.1] ($(code.north west) + (0,-\shift-\X * \lineLen)$) rectangle ($(code.north east) + (0, -\shift)+ (0, -\X * \lineLen - \len * \lineLen + \decr)$);

      \pgfmathsetmacro{\X}{13}
      \pgfmathsetmacro{\len}{2}

      \draw[ao(english)!100, thick, rounded corners, fill=green!10, fill opacity=0.1] ($(code.north west) + (0,-\shift-\X * \lineLen)$) rectangle ($(code.north east) + (0, -\shift)+ (0, -\X * \lineLen - \len * \lineLen + \decr)$);

    \end{tikzpicture}
    }
    \vspace{4em}
  \end{subfigure}
  \caption{CFG of $\func_1$ from Figure~\ref{fig:F1F2seqalign} (left) and the resulting linearization if we traverse $\{v_2,v_3\}$ before $\{v_4\}$ (right).}
  \label{fig:cfgandreorder}
\end{figure}

\paragraph{Reordering Branches} In this sense, the SA-based approach only works on a single linearization of the input functions, which may be suboptimal. A natural question is whether we can do better by considering \emph{multiple} linearizations of the input functions, and selecting the best one for merging. In particular, we consider \emph{reordering the branches} appearing in each function to obtain equivalent but syntactically different functions, and then performing the SA-based merging on the resulting strings. Slightly more formally, our problem is as follows: Given two functions $\func_1, \func_2,$ we want to reorder the branches of their branching constructs (e.g.,~$\texttt{if}/\texttt{switch}$) to obtain equivalent functions $\func_1', \func_2'$ such that $\SA(\func'_1, \func'_2)$ is maximized. The benefit of doing so can be readily seen in our example: The function $\func_1'$ obtained by swapping the two branches in $\func_1$ has a better alignment with $\func_2$ than $\func_1$ itself.  In particular, the entirety of the \texttt{if-elif} block of $\func_1'$ and $\func_2$ can now be aligned, resulting in the merged function $\func_\star'$ shown in Figure~\ref{fig:F1'F2seqalignmerged}, which is 3 lines smaller than the first merging in $\func_\star.$

\paragraph{Branching Factor and Nesting Depth} The \emph{branching factor} of a function $\func$ is the maximum number of branches in any of its branching constructs, and its \emph{nesting depth} is the maximum number of nested branching constructs, e.g.,~the function $\func_1$ in Figure~\ref{fig:F1F2seqalign} has branching factor of $2,$ since it has an \texttt{if-elif} with two branches, and it has nesting depth of $2$ due to lines 9-10 which are enclosed within a \texttt{for} within an \texttt{if}. As our results suggest, these parameters play a key role in the complexity of the problem, and naturally lend themselves to \emph{parameterized} analysis.

\vspace{-0.6em}

\begin{wrapfigure}{r}{0.33\linewidth}\centering\vspace{-1em}
  \hspace{0.5em}\scalebox{0.85}{
  \begin{tikzpicture}
    \node[anchor=north west, inner sep=0] (code) at (0, 0) {
\begin{lstlisting}[firstline=57, lastline=75, style=pythonstyle]
def F$_\star'$(x, b):
  x += 1
  if b == 1:
    y = x * 75
  z = -2 * x
  if b == 1:
    print(y)
  else:
    print(x + z)
  if x == 1:
    print(z)
    print(x * z)
  elif x == 0:
    print(x)
    for i in range(3):
      print(i)
      print(i**2)
  print(x**2)
  print(z**3)
\end{lstlisting}\par
    };
    \pgfmathsetmacro{\lineLen}{0.458}
    \pgfmathsetmacro{\shift}{0.14}
    \pgfmathsetmacro{\decr}{0.07}

    \pgfmathsetmacro{\X}{1}
    \pgfmathsetmacro{\len}{1}

    \draw[ao(english)!100, thick, rounded corners, fill=green!10, fill opacity=0.1] ($(code.north west) + (0,-\shift-\X * \lineLen)$) rectangle ($(code.north east) + (0, -\shift)+ (0, -\X * \lineLen - \len * \lineLen + \decr)$);

    \pgfmathsetmacro{\X}{2}
    \pgfmathsetmacro{\len}{2}

    \draw[bostonuniversityred!100, thick, rounded corners, fill=red!10, fill opacity=0.1] ($(code.north west) + (0,-\shift-\X * \lineLen)$) rectangle ($(code.north east) + (0, -\shift)+ (0, -\X * \lineLen - \len * \lineLen + \decr)$);

    \pgfmathsetmacro{\X}{4}
    \pgfmathsetmacro{\len}{1}

    \draw[ao(english)!100, thick, rounded corners, fill=green!10, fill opacity=0.1] ($(code.north west) + (0,-\shift-\X * \lineLen)$) rectangle ($(code.north east) + (0, -\shift)+ (0, -\X * \lineLen - \len * \lineLen + \decr)$);

    \pgfmathsetmacro{\X}{5}
    \pgfmathsetmacro{\len}{2}

    \draw[bostonuniversityred!100, thick, rounded corners, fill=red!10, fill opacity=0.1] ($(code.north west) + (0,-\shift-\X * \lineLen)$) rectangle ($(code.north east) + (0, -\shift)+ (0, -\X * \lineLen - \len * \lineLen + \decr)$);

    \pgfmathsetmacro{\X}{7}
    \pgfmathsetmacro{\len}{2}
    
    \draw[blue(pigment)!100, thick, rounded corners, fill=blue!10, fill opacity=0.1] ($(code.north west) + (0,-\shift-\X * \lineLen)$) rectangle ($(code.north east) + (0, -\shift)+ (0, -\X * \lineLen - \len * \lineLen + \decr)$);

    \pgfmathsetmacro{\X}{9}
    \pgfmathsetmacro{\len}{8}

    \draw[ao(english)!100, thick, rounded corners, fill=green!10, fill opacity=0.1] ($(code.north west) + (0,-\shift-\X * \lineLen)$) rectangle ($(code.north east) + (0, -\shift)+ (0, -\X * \lineLen - \len * \lineLen + \decr)$);

    \pgfmathsetmacro{\X}{17}
    \pgfmathsetmacro{\len}{2}

    \draw[ao(english)!100, thick, rounded corners, fill=green!10, fill opacity=0.1] ($(code.north west) + (0,-\shift-\X * \lineLen)$) rectangle ($(code.north east) + (0, -\shift)+ (0, -\X * \lineLen - \len * \lineLen + \decr)$);

  \end{tikzpicture}
  }
  \caption{Merged function $\func_\star'$ obtained via the SA-based approach on $\func_1'$ from Figure \ref{fig:cfgandreorder} and $\func_2$ from Figure~\ref{fig:F1F2seqalign}.}
  \vspace{-0.8em}  
  \label{fig:F1'F2seqalignmerged}
\end{wrapfigure}

\paragraph{Parameterized Algorithms and Complexity~\cite{paramcplxty}} In classical complexity theory, we analyze the complexity of a computational problem $\mathcal{A}$ based on its behavior with respect to the \emph{size} of its inputs. That is, given an input $x$ of size $|x|,$ $\mathcal{A}$ is considered efficiently solvable if it has an algorithm with a low dependence (e.g.,~polynomial) on the input size $|x|.$ \emph{Parameterized complexity theory} generalizes this idea by allowing additional parameters to be associated with an input $x.$ It then explores whether problems that are hard on general inputs can be solved by an algorithm that depends polynomially on $|x|,$ but might have a higher dependence on these extra parameters. We call such an algorithm \emph{fixed-parameter tractable} (\FPT). Formal definitions are given in Section~\ref{sec:formulation}. In addition to its theoretical elegance and depth, parameterized complexity offers a practical workaround for problems that are hard in the classical sense: Suppose that $\mathcal{A}$ is \NP-hard, and thus not in \PTIME~unless \PTIME=\NP. If we can identify certain parameters that are typically small on instances of interest in practice, and if we further design an \FPT~algorithm for $\mathcal{A}$ under these parameters, then we get an efficient practical algorithm for $\mathcal{A}.$ This is despite its \NP-hardness on general inputs, and while still providing guarantees on the optimality and running time of the algorithm. Indeed, various \NP-hard problems admit an \FPT~algorithm under many natural parameters~\cite{paramcplxty,paramalgo}.

\paragraph{Contributions} We consider the problem of function merging with branch reordering defined above. Though the state-of-the-art on function merging~\cite{seqalign,ssafuncmerg} mentions this problem as an open question, it has remained unstudied so far. We formulate the problem rigorously in Section~\ref{sec:formulation}, and then we study it through the lens of parameterized algorithms and complexity. We look at two natural parameters: the \emph{branching factor} and \emph{nesting depth} of input functions. Given two functions $\func_1, \func_2$ to merge, let $n_i, b_i, d_i$ denote the size, branching factor, and nesting depth of $\func_i$ for $i\in\{1, 2\}.$ Let $n=\max(n_1, n_2),$ and define $b, d$ similarly. Our results are as follows:
\begin{myitemize}
  \item A simple algorithm running in time $\bigO\left(2^{b_1 d_1 + b_2 d_2} (b_1+b_2)^2(d_1+d_2)\cdot n_1n_2\right)$~(Theorem~\ref{thm:fpt1}). This shows that the problem is in \FPT~when parameterized by $b_1,d_1, b_2, d_2.$ The algorithm is based on a dynamic programming approach and generalizes the classical SA algorithm.
  \item An algorithm running in time $2^{\bigO(b_2d_2+b_1)}(n_1n_2)^3(n_1+n_2)$ (Theorem~\ref{thm:fpt2}). That is, the problem remains in \FPT~even when one of the functions is of unbounded nesting depth. This algorithm relies on reducing the problem to reachability on a weighted pushdown system~\cite{wpds}.
  \item A compelling hardness result, showing that the problem is \NP-hard even when restricted to $d_1=1,b_2 = d_2 = 0$ (Theorem~\ref{thm:dbdhard}). This implies that the parameterization by $d_1, b_2, d_2$ is not in \FPT, unless \PTIME=\NP.
\end{myitemize}
To the best of our knowledge, this is the first systematic study of function merging with branch reordering from an algorithmic or complexity-theoretic perspective.

\paragraph{Related Works on Function Merging} There have been a handful of recent works on function merging~\cite{funcsim,seqalign,ssafuncmerg,hyfm,f3m,multimerge,newfm}. They usually look at the problem more broadly by additionally considering heuristics to pick which functions to merge. Existing compilers, e.g.,~LLVM, were limited to only merging identical functions, and the first work to go beyond that is~\cite{funcsim}. They merge two functions if their CFGs satisfy a strong notion of isomorphism and the total difference of basic-block code in isomorphic nodes is below a certain threshold. This is a much simpler problem than the general graph isomorphism, and can be solved in polynomial time.  The sequence alignment-based approach~\cite{seqalign} looks at one linearization of each of the two functions and computes an alignment between them to profitably guide the merging. A major drawback of~\cite{seqalign} is that it is not able to handle phi-nodes in SSA-form programs. This technical challenge, which is orthogonal to the theory of merging, was addressed in~\cite{ssafuncmerg}, obtaining better size reductions in practice. The work~\cite{hyfm} proposes further optimizations that achieve a faster memory-efficient compilation without significantly compromising the size reduction. \cite{f3m} proposes a hash-based technique to more effectively identify promising function pairs before merging them. All of these works focused on merging two functions at a time. Merging multiple functions $(\geq 3)$ has been explored in~\cite{multimerge}, which builds on top of~\cite{f3m}. Finally, the work~\cite{newfm} proposed a scalable approach to function merging that composes it with function outlining, and works well within a distributed build environment.



\section{Formulation and Preliminaries}\label{sec:formulation}

In this section, we formally define the problem of function merging with branch reordering (Section~\ref{sec:formulation-1}). Then, we provide some background from parameterized complexity (Section~\ref{sec:formulation-2}) and pushdown systems (Section~\ref{sec:formulation-3}) that will be used in subsequent sections.

\subsection{Function Merging with Branch Reordering}\label{sec:formulation-1}

\paragraph{Abstract Syntax} Let $\Sigma$ be an alphabet containing all instructions of interest, e.g.,~arithmetic operations, memory accesses, etc. We consider programs written in an abstract syntax over the extended alphabet $\hat{\Sigma} := \Sigma\cup \left\{\br, [,],\seq\right\}$ defined by the following grammar:
\begin{equation*}
  \begin{split}
    P \rightarrow \sigma \in \Sigma \;\;\Big|\;\; [P \br \dots\br P] \;\;\Big|\;\;  P \seq P.
  \end{split}
\end{equation*}
That is, a program $P$ can either be a single instruction from $\Sigma,$ a \emph{branching construct} $[P_1\br\dots\br P_k]$ with \emph{branches} $P_1,\dots, P_k,$ each of which is a program, or a \emph{sequential composition} of two programs $P_1\seq P_2.$ We note that the number of branches can be arbitrary, but must be finite. Let $\progs(\Sigma)$ denote the set of all programs derived from this grammar. To illustrate, consider the function $\func_1$ from Figure~\ref{fig:F1F2seqalign}. We denote the instruction on line $i$ by $\texttt{i}.$ Then, its representation in our abstract syntax is:
\begin{equation}\label{eqn:prog}
  P_{\texttt{\func}_1} := \texttt{2}\seq\texttt{3}\seq\texttt{4}\seq\texttt{5}\seq[\texttt{7}\seq [\texttt{9}\seq\texttt{10}] \br \texttt{12}\seq\texttt{13}]\seq\texttt{14}\seq\texttt{15}.
\end{equation}

\paragraph{Remark} Though minimal, this syntax is expressive enough to capture the features relevant to our task of reordering branches for optimal merging. Constructs omitted from the syntax are either unaffected by reordering (e.g.,~the function's name and parameters) or reordering has a negligible effect on them (e.g.,~conditions of branching), so they can be left as an implementation detail.

\paragraph{Size, Branching Factor, and Nesting Depth} The \emph{size} of a program $P\in\progs(\Sigma),$ denoted by $|P|,$ is the number of its characters. Its \emph{branching factor} $\brFa(P)$ is the maximum $k$ such that $P$ contains a subexpression of the form $[P_1 \br \dots \br P_k].$ Finally, the \emph{nesting depth}, or simply the \emph{depth} of $P,$ is denoted by $\neDe(P)$ and defined as the maximum number of nested branching constructs in $P,$ e.g.,~$\neDe(\sigma)=0,\neDe([\sigma\br[\sigma]])=2$ where $\sigma\in\Sigma.$ For $P_{\texttt{\func}_1}$ from Equation~(\ref{eqn:prog}), we have $\brFa(P_{\texttt{\func}_1})=\neDe(P_{\texttt{\func}_1})=2,$ which aligns with its original representation in Figure~\ref{fig:F1F2seqalign}~(left).

\paragraph{Reorderings} Given a program $P\in\progs(\Sigma),$ a \emph{reordering} of it is a function $\reOrd$ defined over the subexpressions of $P$ that form a branching construct $Q = [Q_1\br \dots \br Q_k].$ $\reOrd(Q)$ gives an order over the branches of $Q,$ i.e.,~$\reOrd(Q)$ is a permutation $(i_1,\dots, i_k)$ of the set $\{1,\dots, k\}.$ A reordering $\reOrd$ induces a \emph{reordered program} $P^{\reOrd}\in\progs(\Sigma)$ by replacing $[Q_1\br \dots \br Q_k]$ with $[Q_{i_1}\br \dots \br Q_{i_k}]$ for all branching subexpressions $Q.$ The \emph{identity reordering} $\pi_{\text{I}}$ is a reordering that always returns the identity permutation. For example, $P_{\texttt{\func}_1}$ has only two possible reorderings: the identity reordering $\pi_{\text{I}}$ which induces $P_{\texttt{\func}_1}^{\pi_{\text{I}}}=P_{\texttt{\func}_1},$ and the reordering $\pi_2$ that swaps the two subexpressions $\texttt{7}\seq [\texttt{9}\seq\texttt{10}]$ and $\texttt{12}\seq\texttt{13},$ inducing the reordered program $P_{\texttt{\func}_1}^{\pi_2} := \texttt{2}\seq\texttt{3}\seq\texttt{4}\seq\texttt{5}\seq[\texttt{12}\seq\texttt{13} \br \texttt{7}\seq [\texttt{9}\seq\texttt{10}]]\seq\texttt{14}\seq\texttt{15}.$

\paragraph{Linearizations}

Given a reordering $\reOrd$ of $P,$ the \emph{linearization of $P$ under $\reOrd$} is a string $\linearize(P, \reOrd)\in\Sigma^*$\footnote{Recall that $\Sigma^*$ denotes the set of all finite-length strings over alphabet $\Sigma.$} obtained by removing non-$\Sigma$ characters from the reordered program $P^\reOrd.$ It is clear that $|\linearize(P,\reOrd)|\leq |P^\reOrd|=|P|.$ For instance, the linearization of $P_{\texttt{\func}_1}$ under $\pi_2$ is:
\begin{equation}\label{eqn:linearization2}
  \begin{split}
     \linearize(P_{\texttt{\func}_1},\pi_2) = \texttt{2}\;\;\texttt{3}\;\;\texttt{4}\;\;\texttt{5}\;\;\texttt{12}\;\;\texttt{13}\;\;\texttt{7}\;\;\texttt{9}\;\;\texttt{10}\;\;\texttt{14}\;\;\texttt{15}.
  \end{split}
\end{equation}
As argued in Section~\ref{sec:introduction}, a linearization preserves the meaning of $P$ but is stricter in the sense that it fixes the order of its branches according to $\reOrd$ without the possibility of further reordering.

\paragraph{Sequence Alignment and Longest Common Subsequence} Consider two strings $S_1, S_2\in\Sigma^*$ of respective lengths $n_1, n_2,$ and a \emph{scoring function} $\delta : (\Sigma\cup \{-\}) \times (\Sigma \cup \{-\})\rightarrow \mathbb{R}$ that can be evaluated in $\bigO(1)$ time.\footnote{It is more standard that the scoring function is given as a $(|\Sigma|+1) \times (|\Sigma|+1)$ matrix~\cite{bioinfintro}, but since the number of possible instructions can be infinite, we define the scoring as a function and assume oracle access to it.} `$-$' is a special gap character not in $\Sigma.$  An \emph{alignment} of $S_1, S_2$ is a $2\times m$ matrix $M$ where:
\begin{myitemize}
  \item Every row contains an element from $\Sigma\cup \{-\};$
  \item For $i\in\{1, 2\},$ removing all `$-$' characters from the $i$'th row yields $S_i;$ and
  \item No column contains two `$-$' characters.
\end{myitemize}
Note that we necessarily have $\max(n_1, n_2)\leq m\leq n_1 + n_2.$ The \emph{score} of an alignment $M$ under $\delta$ is defined as $\score_\delta(M)=\sum_{j=1}^m \delta(M[1, j], M[2, j]),$ where $M[i, j]$ is the element in the $i$'th row and $j$'th column of $M.$ In the sequence alignment problem, the goal is to compute: $$\SA(S_1, S_2, \delta) := \max_{M}\; \score_\delta(M),$$ where $M$ ranges over all alignments of $S_1, S_2.$ A well-known special case of sequence alignment is the \emph{longest common subsequence} (LCS) problem,\footnote{Recall that for a sequence $a_1, \dots, a_n,$ a \emph{subsequence} of it is any sequence of the form $a_{i_1},\dots a_{i_k}$ where $k\geq 0$ and $i_j<i_{j+1}$ for all $j$ s.t. $1\leq j<k.$} where $\delta$ is defined as $\delta_{\LCS}(a, a) = 1$ for all $a\in\Sigma,$ and $\delta_{\LCS}(a, b) = \delta_{\LCS}(a, -) = \delta_{\LCS}(-, a) = 0$ for all $a, b\in\Sigma$ with $a\neq b.$ In this case, the optimal score, denoted by $\LCS(S_1,S_2) := \SA(S_1, S_2,\delta_{\LCS}),$ is simply the length of the longest common subsequence between $S_1$ and $S_2.$ Both $\SA$ and $\LCS$ can be computed using dynamic programming in time $\bigO(n_1 n_2)$~\cite{bioinfintro,clrs}.

For instance, let $S_1 := \linearize(P_{\texttt{\func}_1},\pi_2)$ (from Equation~(\ref{eqn:linearization2})), and let $S_2:=\linearize(P_{\texttt{\func}_2},\pi_{\text{I}}),$ i.e.,~the result of abstracting $\func_2$ (Figure~\ref{fig:F1F2seqalign} (right)) then linearizing it without reordering. We get the strings:
\begin{equation*}
  \begin{split}
  S_1 = \texttt{2}\;\;\texttt{3}\;\;\texttt{4}\;\;\texttt{5}\;\;\texttt{12}\;\;\texttt{13}\;\;\texttt{7}\;\;\texttt{9}\;\;\texttt{10}\;\;\texttt{14}\;\;\texttt{15},\quad\quad
  S_2 = \texttt{17} \;\; \texttt{18} \;\; \texttt{19}\;\; \texttt{21}\;\; \texttt{22}\;\; \texttt{24}\;\; \texttt{26}\;\; \texttt{27}\;\; \texttt{28}\;\; \texttt{29}.
  \end{split}
\end{equation*}
By replacing line numbers with symbols that identify equivalent instructions, we get:
\begin{equation*}
  \begin{split}
  S_1 = a\;\;b\;\;c\;\;d\;\;h\;\;i\;\;e\;\;f\;\;g\;\;j\;\;k,\quad\quad
  S_2 = a\;\;c\;\;l\;\;h\;\;i\;\;e\;\;f\;\;g\;\;j\;\;k.
  \end{split}
\end{equation*}
The following is an optimal alignment with a score of $9$ under the LCS scoring function $\delta_{\LCS}$:
\begin{equation*}
  \begin{array}{cccccccccccc}
    a & b & c & d & - & h & i & e & f & g & j & k, \\
    a & - & c & - & l & h & i & e & f & g & j & k.
  \end{array}
\end{equation*}
This alignment corresponds to the merged function $\func'_\star$ in Figure~\ref{fig:F1'F2seqalignmerged}. The value $\SA(S_1, S_2,\delta_{\LCS})=\LCS(S_1, S_2)=9$ is precisely the number of matched instructions between $\func'_1$ (i.e.,~$\func_1$ after reordering) and $\func_2,$ excluding the instructions that were abstracted away.

\paragraph{Problem Statement} We are now ready to formally define the problem of function merging with branch reordering. Given two functions with abstract programs $P_1, P_2\in\progs(\Sigma),$ and a scoring function $\delta,$ the goal is to compute:
$$
  \FMBR(P_1, P_2,\delta) := \max_{\reOrd_1, \reOrd_2}\; \SA(\linearize({P_1},{\reOrd_1}),\linearize({P_2},{\reOrd_2}),\delta),
$$
where $\reOrd_1, \reOrd_2$ range over all possible reorderings of $P_1$ and $P_2,$ respectively. For consistency with the complexity-theoretic notions given in Section~\ref{sec:formulation-2} and used in later sections, we also consider the decision variant of \FMBR. There, the input includes an additional integer threshold $\tau,$ and the goal is to decide whether the optimal value $\FMBR(P_1, P_2,\delta)$ is at least $\tau.$ Throughout the remainder of the paper, we let $n_i := |P_i|,$ $b_i := \brFa(P_i),$ and $d_i := \neDe(P_i)$ for $i\in\{1, 2\}.$

\subsection{Background from Parameterized Complexity}\label{sec:formulation-2}

\paragraph{Parameterized Problems} Our definitions follow the standard reference~\cite{paramalgo}. We fix $\Phi$ as our finite alphabet. In classical complexity theory, a problem is defined as a \emph{language} $L\subseteq \Phi^*.$ Then, an algorithm for this problem decides, on a given input $x\in\Phi^*,$ whether $x\in L.$ In parameterized complexity theory, inputs are augmented with one or more \emph{parameters} that capture certain aspects of the input. More formally, a \emph{parameterized problem with $k$ parameters} is a language $L\subseteq \Phi^* \times \mathbb{N}^k.$ An input instance of $L$ is a tuple $\langle x, \param_1,\dots,\param_k\rangle$ where $x\in\Phi^*$ is the main part of the input and $\langle\param_1,\dots, \param_k\rangle\in\mathbb{N}^k$ are the \emph{parameters} associated with $x.$

The standard notion of efficiency in classical complexity is taken to be the class \PTIME, which encompasses all problems with an algorithm running in polynomial time, i.e.,~$\bigO\left(|x|^c\right)$ for some constant $c.$ In the parameterized world, two other complexity classes fulfill this role: \XP~and \FPT.

\paragraph{\XP and \FPT} Let $L$ be a parameterized problem with $k$ parameters. We say that $L\in\XP$ if there is an algorithm that, given an input $\langle x,\param_1,\dots, \param_k\rangle,$ decides whether $\langle x,\param_1,\dots, \param_k \rangle\in L$ in time $\bigO\left(|x|^{f(\param_1,\dots, \param_k)}\right)$ for some computable function $f:\mathbb{N}^k\rightarrow \mathbb{N}.$ We also call such an algorithm an \XP~algorithm. This algorithm is polynomial-time for every fixed value of the parameters, but the degree of the polynomial may depend on the parameters' value. \FPT, short for \emph{fixed-parameter tractable,} is a stricter class than \XP~which requires polynomial input-size dependence regardless of the parameter. We say that $L\in\FPT$ if there is an algorithm that decides whether $\langle x, \param_1,\dots, \param_k\rangle\in L$ in time $\bigO\left(f(\param_1,\dots, \param_k)\cdot |x|^{c}\right)$ for some constant $c$ and a computable function $f:\mathbb{N}^k\rightarrow \mathbb{N}.$ Such an algorithm is also called an \FPT~algorithm. It is immediate that $\FPT\subseteq\XP.$

\paragraph{Parameters of Interest} In our problem $\FMBR,$ the corresponding (unparameterized) language is simply the set $L_{\FMBR} := \{(P_1, P_2,\delta,\tau)\;|\; \FMBR(P_1, P_2,\delta) \geq \tau\}.$ Here, we assume standard binary encoding of the input where programs are represented as a sequence of characters and the scoring function is given as a constant-sized Turing machine. Given an input $x = (P_1, P_2,\delta,\tau),$ we are interested in parameterized variants of $L_{\FMBR}$ obtained by associating with $x$ a non-empty subset of the four parameters $b_1,b_2, d_1,d_2.$ For each parameterization, we can ask whether it is in \FPT~or \XP. For instance, if we choose $b_1,d_2,$ we get the parameterized problem:
$$
  L_{\FMBR}^{b_1,d_2} := \left\{\big\langle (P_1, P_2,\delta,\tau), b_1,d_2\big\rangle\;\big|\; \FMBR(P_1, P_2,\delta) \geq \tau\land b_1=\brFa(P_1)\land d_2=\neDe(P_2)\right\},
$$
which we simply refer to as $\FMBR$ \emph{parameterized by} $b_1$ and $d_2.$ Finally, note that if a parameterized problem with $k$ parameters $\param_1,\dots,\param_k$ is in \FPT~(\XP), then the same problem parameterized by any superset of these parameters is also in \FPT~(\XP).
\subsection{Background from Pushdown Systems}\label{sec:formulation-3}

In this subsection, we provide some relevant background on weighted pushdown systems and reachability over them. Our presentation mostly follows the seminal work of~\cite{wpds}.

\paragraph{Idempotent Semirings} An \emph{idempotent semiring} $\semiring$ is a tuple $\semiring = (\semiringSet, \add, \mult, \zero, \one)$ where $\semiringSet$ is a set, $\zero$ and $\one$ are elements of $\semiringSet,$ and $\add,\mult : \semiringSet^2\rightarrow \semiringSet$ are binary operations over $\semiringSet$ satisfying:
\begin{myenumerate}
  \item $\add$ is commutative, associative, and idempotent, i.e.,~$a\add a = a$ for all $a\in \semiringSet;$
  \item $\mult$ is associative and distributes over $\add;$
  \item $\zero$ is the neutral element for $\add$ and an annihilator for $\mult,$ i.e.,~$a\mult \zero = \zero \mult a = \zero$ for all $a\in \semiringSet;$
  \item $\one$ is the neutral element for $\mult.$
\end{myenumerate}

\paragraph{Weighted Pushdown Systems} A \emph{weighted pushdown system (WPDS) $\pds$ over semiring $\semiring$} is a tuple $\pds = (\pdsstates, \pdsstk, \pdstrans,\wt)$ where:
\begin{myenumerate}
  \item $\pdsstates$ is a finite set of \emph{control locations};
  \item $\pdsstk$ is a finite set called the \emph{stack alphabet};
  \item $\pdstrans\subseteq (\pdsstates\times \pdsstk) \times (\pdsstates\times \pdsstk^*)$ is a finite set of \emph{transitions}, each of which is denoted by $(s, \gamma)\pdsmove(s',w)$ where $s,s'\in\pdsstates,\gamma\in\pdsstk,w\in\pdsstk^*$; and
  \item $\wt : \pdstrans\rightarrow \semiringSet$ is a \emph{weight function} that assigns a weight from $\semiringSet$ to each transition.
\end{myenumerate}
A \emph{configuration} is a pair $\langle s, w\rangle$ where $s\in\pdsstates,w\in\pdsstk^*.$ We call $w$ the \emph{stack} of the configuration, whose first element refers to the top of the stack. The WPDS $\pds$ induces a \emph{configuration graph} $G_\pds$ where:
\begin{myitemize}
  \item The set of nodes is the (infinite) set of all configurations of $\pds;$ and
  \item For every transition $t:(s, \gamma)\pdsmove(s', w)\in\pdstrans$ in $\pds$ and every $w'\in\pdsstk^*,$ we add to $G_\pds$ the edge $e:\left(\langle s, \gamma w'\rangle,\langle s', w w'\rangle\right)$ and we further assign to $e$ a weight $\wt(e):=\wt(t).$ 
\end{myitemize} 

\paragraph{Reachability on WPDSs} Given a WPDS $\pds$ over $\semiring,$ we can now ask reachability questions over its configuration graph $G_\pds.$ Given a (possibly infinite) path of configurations $\rho=c_0,c_1,c_2\dots$ in $G_\pds,$ we define its weight as $\wt(\rho)=\wt(c_0, c_1)\mult\wt(c_1,c_2)\mult\dots,$ where $\wt(\rho)=\one$ if $\rho$ is empty. Given a source configuration $c$ and a target configuration $c',$ we define $\paths(c, c')$ as the set of all paths in $G_\pds$ starting from $c$ and ending at $c'.$ Given $c, c',$ we are interested in computing: 
$$
  \reach(c,c') := \bigoplus_{\rho\in\paths(c, c')} \wt(\rho),
$$ 
where $\oplus(\emptyset)=\zero.$ The work~\cite{wpds} gives an algorithm for computing $\reach(c,c'),$ which we will use in Section~\ref{sec:algorithms-2} to solve $\FMBR.$ 


\section{Algorithms}\label{sec:algorithms}

In this section, we present two \FPT~algorithms for \FMBR. The first one, presented in Section~\ref{sec:algorithms-1}, is for the parameterization by all four parameters $b_1, d_1, b_2, d_2.$ It is based on dynamic programming, and when instantiated with straight-line programs, it coincides with the classical algorithm for \SA/\LCS. Next, in Section~\ref{sec:algorithms-2}, we give an algorithm that allows one of the programs to have unbounded depth, showing that the parameterization by $b_1, b_2, d_i,$ for $i\in\{1, 2\},$ is in \FPT. The main insight behind it is that we can reduce \FMBR~to reachability on a compact WPDS. Combining this with the algorithm of~\cite{wpds}, we obtain an \FPT~algorithm.

\paragraph{Remark} Throughout this section, we use three pseudocodes (Algorithms~\ref{alg:enum-eulerian}-\ref{alg:signature-ops}). We present them modularly to simplify development and avoid repetition, which may make an algorithm seem more complicated than necessary when viewed in isolation.

\subsection{Dynamic Programming Algorithm}\label{sec:algorithms-1}

\paragraph{Source of Hardness} We give an algorithm for \FMBR~parameterized by $b_1, d_1, b_2, d_2.$ First, we give an intuition for the source of hardness in this parameterization. Observe that if our input programs only have sequential composition, then the problem reduces to \SA, which is solvable in $\bigO(n^2)$ time~\cite{bioinfintro,clrs}. On the other hand, if we only have branching, then we readily have an \FPT algorithm: Each of the programs $P_i$ is a branching construct with at most $b_i$ branches, each of which is also a branching-only program, and this recursion continues until depth at most $d_i.$ Thus, we have $\bigO(b_i^{d_i})$ branching constructs in $P_i,$ and we can enumerate all $b_i!^{\bigO(b_i^{d_i})}$ reorderings of each program, and then apply \SA~on each pair of resulting linearizations. This gives us a running time of $b!^{\bigO(b^d)} n^2,$ which is \FPT~with respect to $b_1, d_1, b_2, d_2.$ The hardness comes from the simultaneous presence of both sequencing and branching: If we have $n$ sequential compositions of a program with branching factor $b$ and depth $d,$ then we have $b!^{\bigO(b^d\cdot n)}$ linearizations, which is not in \FPT. We tackle this combinatorial explosion of linearizations through dynamic programming, as described next.

\paragraph{Overview}  For ease of understanding, we present our solution in four steps as follows:
\begin{myenumerate}
  \item We first define a convenient graph representation for a program $P$ in $\progs(\Sigma),$ which we call the \emph{branching graph} $\BG(P).$ For a reordering $\pi$ of $P,$ its linearization $\linearize(P,\pi)$ corresponds to an \emph{Eulerian path} in $\BG(P),$ i.e.,~a path visiting each edge exactly once.
  \item We give a recursive exponential-time algorithm to enumerate all Eulerian paths in $\BG(P),$ thereby capturing all possible linearizations of $P$ induced by reordering its branches.
  \item We show how two simultaneous runs of the above algorithm on branching graphs $\BG(P_1)$ and $\BG(P_2)$ can be combined with the ideas from the classical \SA/\LCS~algorithm to give us an exponential-time algorithm for \FMBR.
  \item Finally, and this is the crux of the algorithm, we carefully analyze the state space and compactly represent it using only $f(b_1, d_1, b_2, d_2)\cdot n_1n_2$ states, yielding an \FPT~algorithm.
\end{myenumerate}

\paragraph{Branching Graphs} As noted above, we define a convenient graph representation for programs, which will make it easier to illustrate our algorithmic ideas. An actual implementation can skip this representation and work directly at the program level, avoiding any overhead associated with generating the graph.  Given a program $P\in \progs(\Sigma),$ its \emph{branching graph} $\BG(P)$ is a directed graph which (i)~has three types of edges: sequential edges, branch-entry edges, and branch-exit edges, (ii)~has two designated nodes which we call the \emph{left} and \emph{right} nodes, and (iii)~each node in $\BG(P)$ has three attributes: a label holding the instruction associated with it, if any, a value denoting its nesting depth, and a \emph{branch ID} which is an identifier for the branch containing the node. Formally, $\BG(P)$ is a tuple $\BG(P) = (V, \Eseq,\Ebren,\Ebrex, l, r,\lab,\neDe,\brID)$ where:
\begin{myenumerate}
  \item $V$ is the set of nodes, which roughly correspond to nodes of the parse tree of $P.$
  \item $\Eseq,\Ebren,\Ebrex \subseteq V\times V$ are the set of sequential, branch-entry, and branch-exit edges, respectively. $\Eseq$ encodes moving from one node to the next using sequential composition, $\Ebren$ encodes moving from a branching construct to the left node of one of its branches, and $\Ebrex$ encodes moving from the right node of a branch back to the enclosing branching construct.
  \item $l, r\in V$ are respectively the left and right nodes, denoting the start and end of $P.$
  \item $\lab : V\rightarrow \Sigma\cup\{\brtext\}$ is the labeling function, which either tells us the instruction associated with a node, or indicates that it corresponds to a branching construct.
  \item $\neDe:V \rightarrow \mathbb{N}$ is the nesting depth for nodes.
  \item $\brID: V\rightarrow V$ gives the ID of the branch that a node belongs to, where a branch is identified by its unique leftmost node.
\end{myenumerate}

We construct the branching graph $\BG(P)$ inductively as follows:

\noindent (1)~\textbf{If $P = \sigma\in\Sigma$:} Let $u$ be a fresh node. We set $V=\{u\},\Eseq=\Ebren=\Ebrex=\emptyset,l = r = u,$ and $\lab(u) = \sigma,\neDe(u) = 0,\brID(u) = u.$

\noindent (2)~\textbf{If $P = [P_1 \br \dots \br P_k]$:} Let $\BG(P_i) = (V_i, \Eseq_i, \Ebren_i, \Ebrex_i, l_i, r_i,\lab_i,\neDe_i,\brID_i)$ for $1\leq i\leq k.$ Let $u$ be a fresh node and define:
\begin{myitemize}
  \item $V=\bigcup_{i=1}^k V_i\cup\{u\},$
  \item $\Eseq=\bigcup_{i=1}^k \Eseq_i, \Ebren=\bigcup_{i=1}^k \Ebren_i \cup \{(u, l_i)~|~ 1\leq i\leq k\}, \Ebrex=\bigcup_{i=1}^k \Ebrex_i \cup \{(r_i, u)~|~ 1\leq i\leq k\},$
  \item $l = r = u,$
  \item $\lab(u) = \brtext$ and $\lab(v) = \lab_i(v)$ for $i\in\{1,\dots, k\}$ and $v\in V_i,$
  \item $\neDe(u) = 0$ and $\neDe(v) = \neDe_i(v) + 1$ for $i\in\{1,\dots, k\}$ and $v\in V_i,$
  \item $\brID(u) = u$ and $\brID(v) = \brID_i(v)$ for $i\in\{1,\dots, k\}$ and $v\in V_i.$
\end{myitemize}

\noindent (3)~\textbf{If $P = P_1\seq P_2$:} Let $\BG(P_i) = (V_i, \Eseq_i, \Ebren_i, \Ebrex_i,l_i, r_i,\lab_i, \neDe_i,\brID_i)$ for $i\in\{1, 2\}.$ Define:
\begin{myitemize}
  \item $V = V_1\cup V_2,\Eseq = \Eseq_1\cup \Eseq_2\cup \{(r_1, l_2)\}, \Ebren = \Ebren_1\cup\Ebren_2, \Ebrex = \Ebrex_1\cup\Ebrex_2,l = l_1, r = r_2,$ and
  \item $\lab = \lab_1\cup\lab_2,\neDe = \neDe_1\cup\neDe_2,\brID(v) = \begin{cases}
      \brID_1(v) & \quad \text{if }v\in V_1,                       \\
      l_1        & \quad \text{if } v\in V_2\land\brID_2(v) = l_2, \\
      \brID_2(v) & \quad \text{otherwise.}
    \end{cases}$
\end{myitemize}
Figure~\ref{fig:branching-graph} illustrates the main aspects of the construction. Sequential edges are shown in solid, branch-entry edges in dashed, and branch-exit edges in dotted lines. When relevant, the label of the node is shown next to it. Note that $\BG(P)$ has at most $2|P|$ nodes and edges, and that $\neDe(P) = \max_{u\in V}\neDe(u),\brFa(P) = \max_{u\in V}\big|\{v\;|\; (u, v)\in\Ebren\}\big|.$ As a concrete example, consider again the program $P_{\texttt{\func}_1}$ from Equation~(\ref{eqn:prog}). Its branching graph $\BG(P_{\texttt{\func}_1})$ is given in Figure~\ref{fig:branching-graph-ex}.

\begin{figure}[h]
  \centering
  \includegraphics[scale=.975]{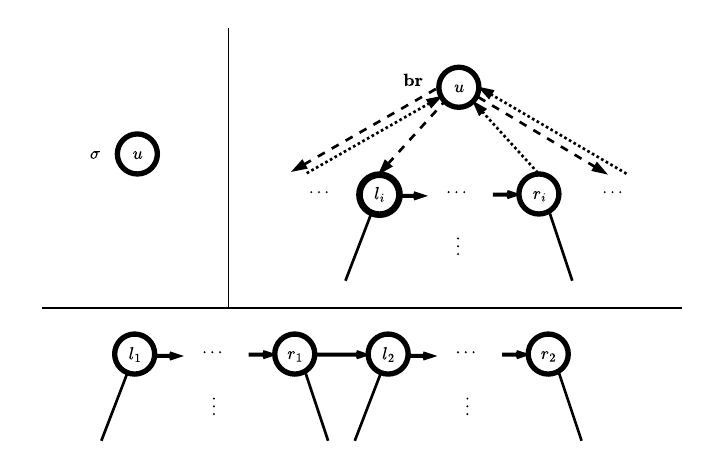}
  \caption{Illustration of the inductive construction of branching graphs. Cases~(1),~(2), and (3) are represented on the top left, top right, and bottom, respectively.}
  \label{fig:branching-graph}
\end{figure}

\begin{figure}[h]
  \centering
  \includegraphics[scale=.975]{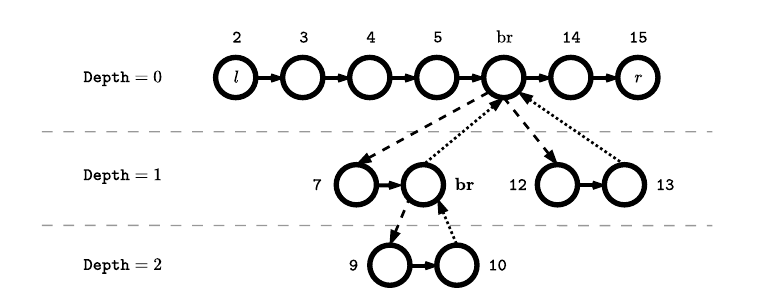}
  \caption{The branching graph corresponding to the program $P_{\texttt{\func}_1}$ in Equation~(\ref{eqn:prog}).}
  \label{fig:branching-graph-ex}
\end{figure}

\paragraph{Linearizations as Eulerian Paths} The intuition behind branching graphs is that, for a linearization $\linearize(P,\pi)$ of $P$ under a reordering $\pi,$ it has a corresponding \emph{Eulerian path} in $\BG(P),$ meaning that it visits every edge of $\BG(P)$ exactly once. Moreover, an Eulerian path in $\BG(P)$ can be uniquely mapped back to a linearization, creating a one-to-one correspondence between linearizations of $P$ and Eulerian paths in $\BG(P).$ 

Given $\linearize(P,\pi),$ we construct the corresponding Eulerian path in $\BG(P)$ as follows. The path starts from the left node $l$ and is extended incrementally. If we are at a node $u$ with label $\sigma\in\Sigma,$ we follow the unique edge outgoing from $u,$ which must be either a sequential or a branch-exit edge. Otherwise, we are at a node $u$ with label $\brtext$ and we have multiple ways of extending the path. Here, we observe that since $\pi$ is a function over branching constructs in $P,$ we can equivalently view it as a function over the $\brtext$-labelled nodes of $\BG(P).$ Given that, we extend the path by visiting the branches of $u$ in the order given by $\pi(u),$ where we initiate the traversal of a branch $P_i$ of $u$ by taking the branch-entry edge $(u, l_i)$ to the left node $l_i$ of $P_i.$ After visiting all branches, we follow the unique edge in $\Eseq\cup\Ebrex$ emitting from $u.$ It is easy to see that this traversal will visit every edge in $\BG(P)$ exactly once and will eventually end at the right node $r,$ which makes it indeed an Eulerian path. It is also easy to see that given an Eulerian path, we can uniquely identify a corresponding linearization of $P$ by simply reading the labels of the $\Sigma$-labelled nodes in the order they are visited along the path. Note that any Eulerian path necessarily starts at $l$ and ends at $r.$

\paragraph{Enumerating Eulerian Paths} Using the insight above, to capture all possible linearizations, it suffices to consider all possible Eulerian paths from $l$ to $r$ in $\BG(P)$ by simply exhausting all possible orders of branches at the $\brtext$-labelled nodes. This is shown in Algorithm~\ref{alg:enum-eulerian}.

\begin{algorithm}[h]
  \caption{Enumerating all linearizations}\label{alg:enum-eulerian}
  \begin{algorithmic}[1]
    \Require A program $P\in\progs(\Sigma).$
    \State Let $\BG(P) = (V, \Eseq, \Ebren, \Ebrex, l, r,\lab, \neDe, \brID).$
    \State \Return \Call{Explore}{$l$, $\emptyset$, $\epsilon$}.

    \State \textbf{//} $u$: currently traversed node; $S$: accumulated linearization; $X$: set of visited nodes.
    \Function{Explore}{$u, X, S$}
    \If{$u = \bot$}\Comment{$\bot$ indicates that we fully traversed the branching graph.}
    \State \Return $\{S\}.$
    \EndIf
    \If{$\lab(u)\in\Sigma$}
    \State $S\gets S\seq \lab(u).$
    \EndIf
    \State \textbf{//} Expanding the traversal.
    \State $Z\gets\{v \;|\; (u, v)\in \Ebren\land v\notin X\}.$\Comment{$Z$: IDs of unvisited branches of $u.$}
    \State \textbf{//} $next$ holds the potential next nodes directly after $u;$ it is never empty if $u\neq \bot.$
    \If{$Z \neq \emptyset$}\Comment{Can only hold if $\lab(u)=\brtext.$}
    \State $next\gets Z.$
    \ElsIf{$u\neq r$}
    \State $next\gets \{v\}$ where $v$ is the unique node such that $(u, v)\in \Eseq\cup\Ebrex.$
    \Else\Comment{It must be that $u = r.$}
    \State $next\gets \{\bot\}.$
    \EndIf
    \State \textbf{//} Updating $X.$
    \State $Y\gets X\cup \{u\}.$
    \State \textbf{//} Continue traversal.
    \State \Return $\bigcup_{v\in next\;} $\Call{Explore}{$v, Y, S$}.
    \EndFunction
  \end{algorithmic}
\end{algorithm}

\begin{lemma}\label{lem:enum-eulerian}
  Given a program $P\in\progs(\Sigma),$ Algorithm~\ref{alg:enum-eulerian} explores all Eulerian paths from $l$ to $r$ in $\BG(P),$ and it returns the set $\big\{\linearize(P,\pi)\;|\;\pi\text{ is a reordering of }P\big\}.$
\end{lemma}

\paragraph{Simultaneous Exploration for Two Programs} We can now run Algorithm~\ref{alg:enum-eulerian} on branching graphs $\BG(P_1)$ and $\BG(P_2)$ to obtain two sets of strings $\textbf{S}_1, \textbf{S}_2.$ This readily gives us an algorithm for \FMBR by simply returning $\max_{S_1\in\textbf{S}_1, S_2\in\textbf{S}_2} \SA(S_1, S_2,\delta).$ Alternatively, we can interleave those steps by simultaneously exploring both branching graphs while computing the SA on-the-fly. This is a key step to enable us to later compact the state space. See Algorithm~\ref{alg:fmbr}. {\scshape Upd-$\ast$} denotes the implementation of all the update operations referenced in {\scshape FuncMerg}. These functions, along with {\scshape Contains}, will be re-defined later to allow us to achieve a better running time.

\begin{restatable}{lemma}{lemcorrectness}\label{lem:correctness}
  Given two programs $P_1, P_2\in\progs(\Sigma)$ and a scoring function $\delta,$ Algorithm~\ref{alg:fmbr} returns $\FMBR(P_1, P_2,\delta).$
\end{restatable}
\begin{proof}
  See Appendix \ref{sec:appsec-1}.
\end{proof}

\begin{algorithm}[h]
  \caption{Computing \FMBR}\label{alg:fmbr}
  \begin{algorithmic}[1]
    \Require Two programs $P_1,P_2\in\progs(\Sigma)$ and a scoring function $\delta.$
    \State Let $\BG(P_i) = (V_i, \Eseq_i, \Ebren_i, \Ebrex_i, l_i, r_i,\lab_i, \neDe_i, \brID_i)$ for $i\in\{1, 2\}.$
    \State \Return \Call{FuncMerg}{$l_1$,$\emptyset$, $l_2$,  $\emptyset$}.

    \State \textbf{//} In $\BG(P_i),$ for $i\in\{1, 2\},$ $u_i$: currently traversed node; $X_i$: set of visited nodes.
    \Function{FuncMerg}{$u_1, X_1, u_2, X_2$}
    \If{$u_1 = \bot\land u_2 = \bot$}\Comment{$\bot$ indicates that we have fully traversed the branching graph.}
    \State \Return $0.$
    \EndIf
    \State \textbf{for} $i\in\{1,2\},$ $Y_i\gets X_i.$\Comment{$Y_i$ holds the updated set of visited nodes after traversing $u_i.$}
    \State \textbf{//} Expanding the traversal.
    \For{$i\in\{1, 2\}$ such that $u_i\neq\bot$}
    \State $Z_i\gets \{v \;|\; (u_i, v)\in \Ebren_i\land \lnot\text{\Call{Contains}{$u_i, v, X_i$}}\}$. \Comment{$Z_i$: IDs of unvisited branches of $u_i$}.
    \State \textbf{//} $next_i$ holds the potential next nodes directly after $u_i;$ it is never empty if $u_i\neq \bot.$
    \If{$Z_i \neq \emptyset$}\Comment{Can only hold if $\lab_i(u_i)=\brtext.$}
    \State $next_i\gets Z_i,Y_i\gets\Call{Upd-en}{i, u_i, X_i}.$
    \ElsIf{$u_i\neq r_i$}
    \State $v\gets$ the unique node such that $(u_i, v)\in \Eseq_i\cup\Ebrex_i.$
    \State $next_i\gets \{v\}.$
    \If{$(u_i, v)\in \Eseq_i$} $Y_i\gets\Call{Upd-seq}{i, u_i, X_i}.$ \Else ~$Y_i\gets\Call{Upd-ex}{i, u_i, X_i}.$ \EndIf
    \Else\Comment{It must be that $u_i = r_i.$}
    \State $next_i\gets \{\bot\},Y_i\gets\Call{Upd-done}{i, u_i, X_i}$
    \EndIf
    \EndFor

    \State \textbf{//} If $u_i$ is a branching construct, then it does not hold an instruction; continue its traversal.
    \If{$u_1\neq\bot\land\lab_1(u_1)=$~`$\brtext$'}
    \State \Return $\max_{w_1\in next_1}\;$\Call{FuncMerg}{$w_1, Y_1, u_2, X_2$}.
    \EndIf
    \If{$u_2\neq\bot\land\lab_2(u_2)=$~`$\brtext$'}
    \State \Return $\max_{w_2\in next_2}\;$\Call{FuncMerg}{$u_1, X_1, w_2, Y_2$}.
    \EndIf
    \State \textbf{//} SA-like transitions.
    \State \textbf{//} Here, either $u_i = \bot,$ $next_i = \{\bot\},$ or $next_i = \{v\}$ where $(u_i, v)\in \Eseq_i\cup\Ebrex_i.$
    \State \textbf{for} $i\in\{1, 2\},$ \textbf{if} $u_i\neq \bot,$ let $w_i$ be the unique element of $next_i.$
    \If{$u_1=\bot$}\Comment{Must have $u_2\neq\bot\land \lab_2(u_2)\in\Sigma.$}
    \State \Return $\delta(-, \lab_2(u_2)) \; + \;$\Call{FuncMerg}{$u_1, X_1, w_2, Y_2$}.
    \EndIf
    \If{$u_2=\bot$}\Comment{Must have $u_1\neq\bot\land \lab_1(u_1)\in\Sigma.$}
    \State \Return $\delta(\lab_1(u_1), -) \; + \; $\Call{FuncMerg}{$w_1, Y_1, u_2, X_2$}.
    \Else\Comment{Must have $\forall i\in\{1, 2\}.\;u_i\neq\bot\land\lab_i(u_i)\in\Sigma.$}
    \State \Return $\max\begin{cases}
        \delta(\lab_1(u_1), \lab_2(u_2)) + \text{\Call{FuncMerg}{$w_1, Y_1, w_2, Y_2$}}, \\
        \delta(\lab_1(u_1), -) + \text{\Call{FuncMerg}{$w_1, Y_1, u_2, X_2$}},           \\
        \delta(-, \lab_2(u_2)) + \text{\Call{FuncMerg}{$u_1, X_1, w_2, Y_2$}}.
      \end{cases}$
    \EndIf
    \EndFunction
    \Function{Upd-${\ast}$}{$i, u, X$}:
    \Return $X\cup \{u\}.$
    \EndFunction
    \Function{Contains}{$u, v, X$}: \Return \textbf{True} if $v\in X,$ and \textbf{False} otherwise.
    \EndFunction
  \end{algorithmic}
\end{algorithm}

\paragraph{Analysis of the State Space} We can apply dynamic programming on Algorithm~\ref{alg:fmbr} through a standard memoization: Whenever we reach a state $(u_1, X_1, u_2, X_2),$ we first check if we have already computed its value; if so, we return the stored value. Given the exponential number of states, this takes exponential time in the worst case. However, we can get a much better upper bound by carefully analyzing the structure of the sets $X_1, X_2.$ Namely, we will show that at a node $u_i$ currently traversed in $\BG(P_i),$ any corresponding set $X_i$ of visited nodes can be characterized by the reordering choices made at the $\bigO(d_i)$ branching constructs enclosing $u_i,$ leading to a limited number of possibilities for $X_i,$ which only depends on $b_i$ and $d_i.$ We now formalize this intuition.

In a program $P$ and a node $u$ in $\BG(P),$ we define its \emph{region} as the set of nodes whose subexpression appears within $u$'s subexpression. Formally, define $Region(u)= \{u\}\cup\bigcup_{(u, v)\in E^{en}}\overline{Region}(v)$ where $\overline{Region}(u)=\{u\}\cup\bigcup_{(u, v)\in E^{en}\cup E^{seq}}\overline{Region}(v).$ Note that this recursive definition eventually terminates as it only traverses edges in $\Ebren\cup\Eseq,$ which are acyclic. Finally, given a branch $B$ identified by a left node $l',$ define the \emph{region of $B$} as $\overline{\reg}(l').$ That is, $\reg(B)$ is union of the regions of all nodes in $B.$

Now suppose that we are traversing a node $u$ in $\BG(P)$ with the set of visited nodes being $X.$ Suppose that $\neDe(u) = d.$ Then, there are $d$ branching constructs enclosing $u$ with corresponding $\brtext$-labelled nodes $v^0,\dots, v^{d-1}$ such that $\neDe(v^i) = i.$ Let $v^d= u.$ For every $i\in\{0, \dots, d\},$ let $l^i, r^i$ be the right and left nodes of the branch of $v^i.$ This situation is illustrated in Figure~\ref{fig:compact}. We have $d+1$ paths $\rho_i=l^i\leadsto v^i\leadsto r^i$ using only $\Eseq$-edges for $i\in\{0, \dots, d\}.$ By definition of how we traverse branching graphs, we observe the following:
\begin{myitemize}
  \item For every $i$ and every node $w$ in the prefix $l^i\leadsto v^i$ with $w\neq v^i,$ we have $\reg(w)\subseteq X.$
  \item For every $i$ and node $w$ in the suffix $v^i\leadsto r^i$ with $w\neq v^i,$ we have $X\cap\reg(w)=\emptyset.$
  \item $\{v^0,\dots, v^{d-1}\}\subseteq X,$ and either $\reg(v^d)\cap X=\emptyset$ or $\reg(v^d)\subseteq X.$
  \item For every $i,$ suppose $v^i$ has $k_i$ branches $B_1^i,\dots, B_{k_i}^i$ with left nodes $l_1^i,\dots, l_{k_i}^i.$ Then, for every $j\in\{1,\dots, k_i\}$ s.t. $v^d\notin\reg(B_j^i),$ either $\reg(B_j^i)\subseteq X$ or $X\cap \reg(B_j^i) = \emptyset.$
\end{myitemize}
Based on these observations, it is now clear that the different possible values for $X$ when traversing $u$ are completely determined by the choice of which branches to include from each branching construct $v^i,$ for $i\in\{0, \dots, d\}.$\footnote{If $u=v^d$ is $\Sigma$-labelled, then we treat it as a branching construct with 0 branches.} Since each branching construct has at most $\brFa(P)$ branches, there are at most $2^{\brFa(P)}$ choices per branching construct. At the same time, either $\lab(u)=\brtext$ and we need to make $d+1\leq\neDe(P)$ choices at $v^0,\dots, v^d,$ or $\lab(u)\in\Sigma$ and we need to make $d\leq\neDe(P)$ choices at $v^0,\dots, v^{d-1},$ because there is only a single choice at $v^d.$ In either case, we have at most $2^{\brFa(P)\cdot \neDe(P)}$ possible values for $X$ when traversing $u.$ 

With this analysis, we get that throughout the execution of Algorithm~\ref{alg:fmbr}, for a fixed node $u_i,$ there are at most $2^{b_id_i}$ possible values for $X_i.$ As we remember $u_i, X_i$ for each of the two programs, we get the following:

\begin{figure}[h]
  \centering
  \includegraphics[scale=.975]{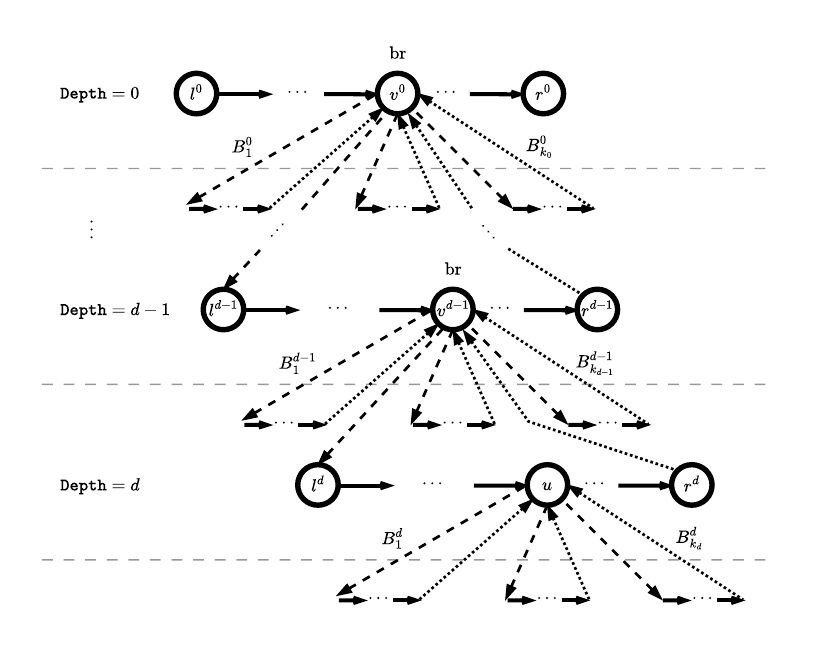}
  \caption{An illustration of the shape of visited nodes $X$ when traversing a node $u.$}
  \label{fig:compact}
\end{figure}

\begin{lemma}\label{lem:state-space}
  Given two programs $P_1, P_2\in\progs(\Sigma)$ and a scoring function $\delta,$ the number of distinct states $(u_1, X_1, u_2, X_2)$ encountered while executing Algorithm~\ref{alg:fmbr} is at most $\bigO\left(2^{b_1 d_1 + b_2 d_2}\cdot n_1n_2 \right)$.
\end{lemma}

\paragraph{Signatures} Lemma~\ref{lem:state-space} makes Algorithm~\ref{alg:fmbr} with memoization an \FPT~algorithm. However, manipulating and passing the sets $X_1, X_2$ throughout the recursive calls incurs an additional cost of $\bigO(n_1+n_2),$ leading to a costly input dependence of $\bigO(n_1n_2(n_1+n_2)).$ We avoid this by representing $X_i$ compactly via a \emph{signature.} Consider again the situation in Figure~\ref{fig:compact} when we traverse a node $u$ of $\BG(P).$ A \emph{signature} for $u$ is a \emph{stack} of $d+1$ records $[S_d, \dots, S_0]$ where each $S_i$ is a subset of $\{1,\dots, k_i\}$ indicating which branches of $v^i$ have been fully traversed. That is, $j\in S_i$ iff $\reg(B_j^i)\subseteq X.$ Here, the leftmost record $S_d$ is the top record of the stack. We now modify the algorithm to let $X_i$ denote a signature instead of a subset of $V_i,$ and change the {\scshape Upd-$\ast$} and {\scshape Contains} functions accordingly as shown in Algorithm~\ref{alg:signature-ops}. Finally, we replace the initial call to {\scshape FuncMerg} in line 3 to pass $[\emptyset],$ i.e.,~a signature containing one empty record, for $X_1$ and $X_2.$ With this change, the work we do in every state is only polynomial in $b_1, b_2, d_1, d_2,$ giving us the following theorem, which establishes that the parameterization by $b_1, b_2, d_1, d_2$ is in \FPT:

\begin{restatable}{theorem}{thmfptone}\label{thm:fpt1}
  Given two programs $P_1, P_2\in\progs(\Sigma)$ and a scoring function $\delta,$ the modified algorithm described above returns $\FMBR(P_1, P_2,\delta)$ and runs in time $\bigO\left(2^{b_1 d_1 + b_2 d_2} (b_1+b_2)^2(d_1+d_2)\cdot n_1n_2\right).$
\end{restatable}
\begin{proof}
  See Appendix \ref{sec:appsec-2}.
\end{proof}

\begin{algorithm}[h]
  \caption{Operations on signatures}\label{alg:signature-ops}
  \begin{algorithmic}[1]
    \Function{Upd-en}{$i, u, X$}: $X.push(\emptyset)$; \Return $X.$\EndFunction
    \Function{Upd-seq}{$i, u, X$}: $X.pop, X.push(\emptyset)$; \Return $X.$\EndFunction
    \Function{Upd-ex}{$i, u, X$}:
    \State $X.pop$
    \State $v\gets$ the unique node such that $(u, v)\in\Ebrex_i.$
    \State Let $j$ be s.t. $\brID(u)$ is the ID of the $j$-th branch of $v.$
    \State $X.top\gets X.top\cup \{j\}$\Comment{Adding $j$ to the top record of the stack.}
    \State \Return $X.$
    \EndFunction
    \Function{Upd-done}{$i, u, X$}: \Comment{Must have $u=r_i\land|X|=1.$}
    \State $X.pop$; \Return $X.$\Comment{Now $X=\emptysgn.$}
    \EndFunction
    \Function{Contains}{$u, v, X$}
    \State Let $j$ be s.t. $v$ is the ID of the $j$-th branch of $u.$
    \State \Return \textbf{True} if $j\in X.top,$ and \textbf{False} otherwise.
    \EndFunction
  \end{algorithmic}
\end{algorithm}

\subsection{WPDS-based Algorithm}\label{sec:algorithms-2}

We further modify the algorithm of~Theorem~\ref{thm:fpt1} to avoid the high dependence on one of the depth parameters $d_i,$ resulting in an \FPT algorithm w.r.t. $b_1, b_2, d_{3-i}$ only. The idea is to first view the previous solution as solving a reachability problem on a stack-less WPDS, i.e.,~a weighted finite state machine. Then, we exploit the stack-like structure of signatures to encode one of the two signatures on the stack of the WPDS, subsequently allowing for arbitrary depth over that signature. In this subsection, when we refer to Algorithm~\ref{alg:fmbr}, we mean the modified version with signatures as described above.

\paragraph{Bounded SA/LCS Values} Our algorithm assumes that the scoring function $\delta$ either returns small integer values, or a very large negative value to forbid certain pairs in $(\Sigma\cup\{-\})^2$ from being aligned. This does not impose a restriction in practice as alignment scores for realistic programs are based on a simple scoring scheme where equivalent pairs of instructions get a small non-negative score, and other cases are either forbidden or incur a small penalty~\cite{seqalign,ssafuncmerg}. Thus, throughout the rest of this section, we suppose that $\delta(a,b)\in\{L,\dots, R\}\cup\{-\infty\}$ for all $a,b\in\Sigma\cup\{-\},$ where $L, R$ are integers in $\bigO(1).$ Without loss of generality, we further suppose that $L\leq 0\land 1\leq R.$ Note that $\delta_\LCS$ has $L=0,R=1.$

\paragraph{Intuition} The main insight behind our WPDS formulation is to look at the graph of the recursive calls induced by Algorithm~\ref{alg:fmbr}. Every node corresponds to a dynamic programming state $(u_1, X_1, u_2, X_2)$ and has an associated optimal value, i.e.,~the value {\scshape FuncMerg}$(u_1, X_1, u_2, X_2).$ The crucial observation is that a (possibly suboptimal) solution for this state corresponds to a path from this state to the terminal state $(\bot, \emptysgn, \bot, \emptysgn)$ in the graph of recursive calls, and the value of that solution is the \emph{sum} of the scores accumulated along the edges (choices) in that path. Thus, finding the optimal \FMBR solution boils down to finding the best path from $(l_1, [\emptyset], l_2, [\emptyset])$ to $(\bot, \emptysgn, \bot, \emptysgn)$  achieving the \emph{maximum} value. This motivates the choice of the semiring. 

Formally, we use the semiring $\semiring=(\semiringSet,\add,\mult,\zero,\one)$ where
$\semiringSet=\mathbb{Z}\cup\{-\infty\},\add = \max,\mult=+.$\footnote{Here, it is implicit that for $a\in\mathbb{Z}\cup\{-\infty\},$ we have $a+(-\infty)=(-\infty)+a=-\infty$ and $\max(a,-\infty)=\max(-\infty,a)=a.$} We define $\zero = -\infty,\one = 0.$ It is easy to verify that $\semiring$ is an idempotent semiring as defined in Section~\ref{sec:formulation-3}.

\paragraph{Stack-less WPDS} We now define a \emph{stack-less WPDS} $\pds_1$ where we have a single stack symbol $ \$ $ and all transitions are of the form $(s, \$)\pdsmove(s',\$).$ Each node corresponds to a dynamic programming state and transitions encode the recursive calls of the algorithm. More formally, we define a WPDS $\pds_1=(\pdsstates_1,\{ \$ \}, \pdstrans_1,\wt_1)$ over $\semiring$ where:
  \begin{myitemize}
    \item $\pdsstates_1 = \big\{(u_1, X_1, u_2, X_2) \;|\; \forall i\in\{1, 2\}.\; u_i\in V_i\land X_i\text{ is a signature for } u_i\big\},$ which corresponds to all the dynamic programming states.
    \item For every $s = (u_1, X_1, u_2, X_2)$ where $u_1\neq\bot\lor u_2\neq\bot,$ we look at the execution of {\scshape FuncMerg}$(s).$ It is not hard to see that it always returns a value of the form:
    $$\max\begin{cases}
        \mu_1 + \text{\Call{FuncMerg}{$s_1$}}, \\
        \quad\quad\quad\quad\vdots                    \\
        \mu_k  + \text{\Call{FuncMerg}{$s_k$}},
      \end{cases}$$
    where $k\geq 1$ and $\mu_i\in\{L,\dots, R\}\cup\{-\infty\}.$ Then, for every $i\in\{1,\dots, k\},$ we add to $\pdstrans_1$ a transition $(s, \$)\pdsmove(s_i, \$)$ and set $\wt_1((s, \$)\pdsmove(s_i, \$)) =  \mu_i.$
  \end{myitemize}

  \begin{restatable}{lemma}{lempdsreach}\label{lem:pds-reach}
    Given two programs $P_1, P_2\in\progs(\Sigma)$ and a scoring function $\delta,$ in the WPDS $\pds_1$ over $\semiring$ defined above, we have $\FMBR(P_1, P_2,\delta)=\reach(\langle (l_1, [\emptyset], l_2, [\emptyset]), \$ \rangle,\langle(\bot, \emptysgn, \bot, \emptysgn),\$\rangle).$
  \end{restatable}
  \begin{proof}
    See Appendix \ref{sec:appsec-3}.
  \end{proof}

  \paragraph{Encoding Signatures on the Stack} We now turn the above stack-less WPDS $\pds_1$ into another one $\pds_2$ which only remembers the signature for one of the two programs, say $P_2,$ and uses the pushdown behavior of WPDSs to encode the signature for $P_1.$ 
  Concretely, we define $\pds_2=(\pdsstates_2,\pdsstk_2,\pdstrans_2,\wt_2)$:
  \begin{myitemize}
    \item $\pdsstates_2 = \big\{(u_1, u_2, X_2),(u_1, u_2, X_2)_j \;|\;  u_1\in V_1\land u_2\in V_2\land X_2\text{ is a signature for } u_2\land j\in\{1,\dots, b_1\}\big\}.$
    \item $\pdsstk_2 = \{S \;|\; S\subseteq \{1,\dots,b_1\}\}\cup\{\$\}.$
      \item For every $s = (u_1, X_1, u_2, X_2)$ where $u_1\neq\bot\lor u_2\neq\bot,$ we have two cases:
      
      \textbf{If $u_1=\bot$:} Then we look at the execution of {\scshape FuncMerg}$(\bot,\emptysgn,u_2,X_2).$ Similar to before, we obtain a set $\{(\mu_i, u_1^i, X_1^i, u_2^i, X_2^i)\}_i$ corresponding to successor states. Note that in this case, we are guaranteed to have, for all $i,$ $u_1^i=\bot\land X_1^i=\emptysgn.$ For all $i,$ we add to $\pdstrans_2$ a transition $((\bot, u_2, X_2), \$)\pdsmove((\bot, u_2^i, X_2^i), \$)$ with weight $\mu_i.$

      \textbf{If $u_1\neq\bot$:} Let $\alpha$ be the number of branches of $u_1.$ We consider every set $S\in\pdsstk_2\backslash\{ \$ \}$ where $S\subseteq\{1,\dots,\alpha\}.$ If $\alpha=0,$ then only $S=\emptyset$ is considered. We further distinguish two cases:
    \begin{myitemize}
      \item If there is no $v$ s.t. $(u_1,v)\in \Ebrex_1$ or $S\subset\{1,\dots, \alpha\},$ then trace {\scshape FuncMerg}$(u_1,[S],u_2, X_2)$ to get a set $\{(\mu_i, u_1^i, X_1^i, u_2^i, X_2^i)\}_i.$ We claim that this is a sound way of invoking {\scshape FuncMerg}. Note that by the conditions assumed on $S,$ {\scshape Upd-ex}$(1, u_1, [S])$ will not be called. All other signature operations only need to see the top record of the signature, which, for $u_1,$ is precisely $S.$ Now observe that for every $i,$ we can view $X_1^i$ as a string $T^i\in(\pdsstk_2)^*$ with its first character being the top record. This string tells us precisely how to update a signature with top record $S,$ and thus we simply add the transition $((u_1, u_2, X_2), S)\pdsmove((u_1^i, u_2^i, X_2^i), T^i)$ with weight $\mu_i.$
      \item Otherwise, we need special care to encode the effect of calling {\scshape Upd-ex}. We look at the execution of {\scshape FuncMerg}$(u_1,[S,\emptyset],u_2, X_2)$ to get $\{(\mu_i, u_1^i, X_1^i, u_2^i, X_2^i)\}_i,$ which is again sound as {\scshape Upd-ex} only needs to access the two topmost records. For every $i,$ we must have $X_1^i=\text{\scshape Upd-ex}(1, u_1, [S, \emptyset]) = [\{j\}]$ for some $j\in\{1,\dots, b_1\}.$ We add a transition $((u_1, u_2, X_2), S)\pdsmove((u_1^i, u_2^i, X_2^i)_j, \epsilon)$ with weight $\mu_i.$ Further, let $v$ be s.t. $(u_1, v)\in\Ebrex_1$ and let $\beta$ be the number of branches of $v.$ For every $S'\subseteq\{1,\dots, \beta\}$ with $j\notin S',$ we add a transition $((u_1^i, u_2^i, X_2^i)_j, S')\pdsmove((u_1^i, u_2^i, X_2^i), S'\cup\{j\})$ with weight $0.$
    \end{myitemize}

  \end{myitemize}

  This construction is the basis of the following theorem which shows that \FMBR~is in \FPT~w.r.t. $b_1, b_2, d_2.$ Clearly, this also gives an \FPT~algorithm w.r.t. $b_1, b_2, d_1,$ by swapping the two programs. The proof makes use of the boundedness of score values to show that application of the algorithm of~\cite{wpds} would only take polynomial time. The exponential dependence in the theorem follows in a manner similar to Theorem~\ref{thm:fpt1}, and the extra polynomial dependence on $n_1, n_2$ is due to the overhead of applying the WPDS algorithm.

\begin{restatable}{theorem}{thmfpttwo}\label{thm:fpt2}
  Given two programs $P_1, P_2\in\progs(\Sigma)$ and a scoring function $\delta$ returning values in $\{L,\dots, R\}$ where $L,R\in\bigO(1)\land L < 1\leq R,$ we can compute $\FMBR(P_1, P_2,\delta)$ in time $$2^{\bigO(b_2d_2+b_1)}(n_1n_2)^3(n_1+n_2).$$
\end{restatable}

\begin{proof}
  See Appendix~\ref{sec:appsec-4}.
\end{proof}

\section{Complexity}\label{sec:hardness}

In this section, we study the complexity of \FMBR.

\paragraph{Program Size Dependence} Note that our dynamic programming algorithm in Section~\ref{sec:algorithms-1} has an $\bigO(n_1 n_2)$ factor in its running time. It is natural to ask whether this quadratic dependence can be improved to $\bigO\left(f(b_1,b_2,d_1,d_2)\cdot(n_1 n_2)^{1-\epsilon}\right)$ for some $\epsilon > 0,$ perhaps at the cost of a large $f.$ Unfortunately, this is impossible under the Strong Exponential Time Hypothesis (SETH) due to existing lower bounds for the computation of \LCS~\cite{lcshard1,lcshard2}, and the fact that \FMBR~is a generalization of \LCS. Namely,~\cite{lcshard1,lcshard2} prove that under SETH, computing $\LCS(S_1, S_2)$ on two strings of size $n$ cannot be done in time $\bigO(n^{2-\epsilon})$ for any $\epsilon > 0.$ Since \LCS~is a special case of \FMBR~where both input programs have no branching (i.e.,~for $i\in\{1, 2\},\;b_i=d_i=0$), we have the following corollary:

\begin{corollary}
  Unless SETH fails, there is no algorithm for $\FMBR$ that runs in time $$\bigO\left(f(b_1, b_2, d_1, d_2)\cdot (n_1 n_2)^{1-\epsilon}\right)$$ for an arbitrary computable function $f$ and any $\epsilon > 0.$
\end{corollary}

We now consider lower bounds for parameterized variants of the problem. In the previous section, we showed that \FMBR parameterized by $b_1,b_2,d_i$ is in \FPT~for $i\in\{1, 2\},$ which raises the question of whether we can get rid of more parameters whilst keeping the problem in \FPT, or at least in \XP. In this direction, we prove a negative result: Relaxing the branching factor of one of the programs, while restricting all other three parameters, makes the problem \NP-hard, even with the \LCS~scoring function. This implies that any parameterization not including $b_i,$ for some $i\in\{1, 2\},$ is not in \XP, and thus not in \FPT, unless \PTIME= \NP.

\begin{theorem}\label{thm:dbdhard}
  For $i\in\{1, 2\},$ \FMBR~is \NP-hard even when restricted to inputs where $P_i$ has depth 1, i.e.,~$d_i=1,$ $P_{3-i}$ has no branching at all, i.e.,~$b_{3 - i} = d_{3 - i} = 0,$ and the scoring function is $\delta_\LCS.$
\end{theorem}

\begin{restatable}{corollary}{cornotxp}\label{cor:notxp}
  For $i\in\{1, 2\},$ \FMBR~parameterized by $d_i, b_{3-i}, d_{3-i}$ is not in \XP~and thus not \FPT, unless \PTIME= \NP.
\end{restatable}
\begin{proof}
  See Appendix~\ref{sec:appsec-5}.
\end{proof}

The remainder of this section is dedicated to proving Theorem~\ref{thm:dbdhard}. The proof is by a reduction from the \emph{exact cover by 3-sets} problem, denoted by \XThC, which is known to be \NP-hard~\cite{nphard}. Our~reduction will produce an \FMBR instance where $b_1 = 1,b_2 = d_2 = 0,\delta=\delta_\LCS,$ which proves the theorem for the case $i=1.$ The case $i=2$ trivially follows by modifying the reduction to swap the two programs.

\paragraph{\XThC} In \XThC, we are given a universe $U = \{u_1,\dots, u_{3n}\}$ of size $3n,$ and a family of 3-sets $\mathcal{Q} = \{Q_1, \dots, Q_m\}, \forall i.\;|Q_i| = 3.$ The task is to find $n$ sets $Q_{p_1},\dots, Q_{p_n}$ in that family that cover $U,$ i.e.,~$\bigcup_{i = 1}^n Q_{p_i} = U.$ Note that this necessarily implies that the $Q_{p_i}$'s are mutually disjoint. We assume that $m > n$ as otherwise the problem would be trivial.

\paragraph{Reduction to \FMBR} In the reduced instance, we use alphabet $\Sigma := \{u_1,\dots, u_{3n}, *, Y, Z\}.$ We suppose that for any set $Q_i,$ its elements have a fixed order. When writing $Q_i$ in a program (string), we are referring to the sequential composition (concatenation) of its elements according to this order. The reduced instance contains the following programs:

\begin{align*}
  P_1  :=  [u_1 \br \dots \br u_{3n} \br \underbrace{Y\seq * \br Y\seq * \br\dots\br Y\seq *}_{n\text{ times}}\br \underbrace{Y\seq Z \br Y\seq Z \br\dots\br Y\seq Z}_{m-n\text{ times}}]\seq Y \\
  P_2  := Y \seq H_1\seq\;\; \dots \;\;\seq H_m\quad\quad\quad\quad\quad\quad\quad\quad H_i := *\seq Q_i\seq Z\seq Y.\quad\quad\;\;
\end{align*}

We use the \LCS~scoring function $\delta_\LCS$ and use $\tau_\star := 3n + 2m + 1.$ The reduced instance is $(P_1, P_2,\delta_\LCS, \tau_\star).$ Clearly, the reduction runs in polynomial time.
The following lemma establishes the correctness of the reduction, thus concluding the proof of Theorem~\ref{thm:dbdhard}.
\begin{lemma}
  The \XThC instance has a solution iff $\max_{\pi_1,\pi_2} \LCS(\linearize(P_1,{\pi_1}),\linearize(P_2,{\pi_2}))\geq \tau_\star.$
\end{lemma}
\begin{proof}  
  As $P_1$ has a single branching construct $[B_{1}\br\dots \br B_{3n+m}]$ with $3n+m$ branches, reorderings of $P_1$ correspond to permuting those branches. Let $\textbf{B} = (\linearize(B_{1},\reOrd_{\text{I}}),\dots, \linearize(B_{3n+m},\reOrd_{\text{I}})).$ For a permutation $\pi$ of $\textbf{B},$ let $S_1^\pi$ be the string obtained by permuting $\textbf{B}$ according to $\pi,$ concatenating its elements, then appending the character $Y.$ Let $S_2 :=\linearize(P_2,\reOrd_{\text{I}}).$ Since $P_2$ has no branching, and thus only has the identity reordering $\pi_\text{I},$ we get:
  $$\max_{\pi_1,\pi_2} \LCS(\linearize(P_1,{\pi_1}),\linearize(P_2,{\pi_2})) = \max_{\pi\text{ permutation of }\textbf{B}} \LCS(S_1^\pi, S_2),$$

  and therefore it suffices to show that $\XThC$ has a solution iff $\max_{\pi} \LCS(S_1^\pi, S_2)\geq\tau_\star.$ Observe that any $S_1^\pi$ has exactly $3n+2m+1=\tau_\star$ characters and hence $\max_\pi\LCS(S_1^\pi, S_2)\geq\tau_\star$ holds iff there is a permutation $\pi$ and an alignment between $S_1^\pi$ and $S_2$ where \emph{every} character $\sigma_1$ in $S_1^\pi$ is aligned with a character $\sigma_2$ in $S_2$ in the alignment matrix. In this case we say that $\sigma_1$ and $\sigma_2$ are \emph{matched} and that $S_1^\pi$ is \emph{fully matched} with $S_2.$ We now prove the two directions of the lemma.

  \paragraph{$\implies$} Suppose there is a cover $Q_{p_1},\dots, Q_{p_n}$ of $n$ mutually disjoint sets, where $p_1< \dots < p_n.$ We will construct a string $S_1^\pi,$ for some permutation $\pi$ of $\textbf{B},$ that is fully matched with $S_2.$ This is done in $m+1$ steps. Let $S_1^{(i)}$ be the string constructed after the $i$'th step, where we start with $S_1^{(0)}$ being the empty string. In the $i$'th step, $i\in\{1\dots, m\},$ if $i=p_j$ for some $j\in\{1,\dots, n\},$ then we set $S_1^{(i)} = S_1^{(i-1)} Y*Q_{p_j},$ and otherwise we set $S_1^{(i)} = S_1^{(i-1)} YZ.$ After those $m$ steps, we finally set $S_1^{(m+1)} = S_1^{(m)} Y,$ which adds the last $Y$ in $P_1$ to the string. We now claim that $S_1^{(m+1)}$ corresponds to a string of the form $S_1^\pi$ for some permutation $\pi$ of $\textbf{B}.$ Each of the first $m$ steps can be seen as \emph{removing} some strings from $\textbf{B},$ and throughout those steps, we remove exactly $n$ $Y*$'s, $m-n$ $YZ$'s, and we also remove every $u_i$ exactly once, which follows from the fact that the $Q_{p_i}$'s partition $U.$ Thus, after $m$ steps, we removed precisely all strings in $\textbf{B},$ and $\pi$ can be constructed from the removal order. With the last step, by definition we get $S_1^{(m+1)} = S_1^\pi.$
  
  We show that $S_1^\pi$ is fully matched with $S_2$ inductively by maintaining the following invariant for $i\in\{1, \dots, m\}$: All letters of $S_1^{(i)}$ are matched with $S_2,$ and the last letter in $S_1^{(i)}$ is matched with a letter strictly before the $Y$ in $H_i.$ For $i=1$ regardless of whether $S_1^{(1)} = Y*Q_{1}$ or $S_1^{(1)} = YZ,$ we can fully match it with the $Y*Q_1Z$ prefix of $S_2,$ which makes the invariant hold initially. For $i\in\{2,\dots, m\},$ we have by the invariant that $S_1^{(i-1)}$ is fully matched with $S_2,$ with last character being matched strictly before the $Y$ in $H_{i-1}.$ If $i=p_j$ for some $j\in\{1,\dots, n\},$ then $S_1^{(i)} = S_1^{(i-1)} Y*Q_{p_j}.$ Then, we match $S_1^{(i-1)}$ as before and in the $Y*Q_{p_j}$ suffix of $S_1^{(i)},$ we match the $Y$ with its counterpart in $H_{i-1}$ and match the $*Q_{p_j}$ with its counterpart in $H_i.$ In the other case, we have $S_1^{(i)} = S_1^{(i-1)} YZ.$ Again, we match the $S_1^{(i-1)}$ as before, match the last $Y$ in $S_1^{(i)}$ with the $Y$ in $H_{i-1},$ and finally match the last $Z$ in $S_1^{(i)}$ with the $Z$ in $H_i.$ In both cases, the invariant continues to hold. After we have matched all of $S_1^{(m)}$ with $S_2,$ the invariant tells us that the $Y$ in $H_m$ is still available for matching. We match it with the last $Y$ in $S_1^\pi = S_1^{(m)} Y,$ and now we have a full matching.

  \paragraph{$\impliedby$} For the other direction, suppose that there is a permutation of $\textbf{B}$ s.t. $S_1^\pi$ is fully matched with $S_2.$ We show that this gives rise to a valid solution to the $\XThC$ instance. The concern here is that $S_1^\pi$ tries to cover $U$ inconsistently by looking at more than $n$ sets while skipping some of the elements of those sets. This is where the $Z$'s come into play: They ensure that we match the universe's elements $u_1,\dots, u_{3n}$ in $S_1^\pi$ with exactly $n$ sets in $S_2,$ and for the other $m-n$ sets, represented by $YZ$'s, we are forced to entirely skip/exclude them.

  First, observe that $S_1^\pi$ and $S_2$ have the same number of $Y$'s ($m+1$), and since $S_1^\pi$ is fully matched with $S_2,$ \emph{the matching between the $Y$'s is determined.} This implies that $S_1^\pi$ must start with a $Y*$ or $YZ$ because otherwise we would miss the chance to match the first $Y$ in $S_2.$ Then, we can write $S_1^\pi$ in the form
  $
    V_1 R_1 V_2 \dots V_m R_m Y
  $
  where $V_i \in \{Y*, YZ\}$ with exactly $n$ of the $V_i$'s being $Y*$'s, and the $R_{i}$'s form a partition of $U$ with some $R_i$'s being possibly empty. We will show for any $i\in\{1,\dots, m\},$ any matching of $V_1R_1\dots V_{i} R_{i},$ satisfies the following:
  \begin{myenumerate}
    \item Its last character is matched with a character lying in $H_i$ and strictly before the $Y$ in $H_i;$ and
    \item Either $V_i = Y*$ and $R_i \subseteq Q_i,$ or $V_i = YZ$ and $R_i = \emptyset.$
  \end{myenumerate}
  To see this, let $V_iR_i = YxR_i,$ where $i\in\{1,\dots,m\}$ and $x\in\{*, Z\}.$ Using the fact that the matching for $Y$'s is determined, we get that the $Y$ in $V_i$ must be matched with the first $Y$ in $S_2$ (if $i=1$) or with the $Y$ in $H_{i-1}$ (if $i>1$). At the same time, the $Y$ right after $R_i$ (it always exists) must be matched with the $Y$ in $H_i,$ implying that $xR_i$ is forced to be matched with the $*Q_iZ$ substring of $H_i,$ which already proves the first part of the statement. With this restriction, $x$ must be matched with the corresponding symbol in $H_i$ and $R_i$ must be matched with the set $Q_i$ in $H_i,$ which shows that $R_i\subseteq Q_i.$ It remains to show that if $V_i = YZ$ (i.e.,~$x=Z$), then $R_i=\emptyset.$ Indeed, this is the case because the $Z$ in $V_i$ is matched strictly after the $Q_i,$ and hence we cannot have a non-empty matching from $R_i$ at all. Since $S_1^\pi$ is fully matched, it must be that $R_i=\emptyset.$

  The above statement implies that there are at most $n$ $R_i$'s that are non-empty, and they correspond to subsets of sets from $\mathcal{Q}.$ Since $R_i$'s partition $U,$ those $n$ $R_i$'s must actually represent full sets from $\mathcal{Q}$ because otherwise they will miss some elements from $U;$ recall that each set has exactly $3$ elements. Hence, the $R_i$'s give us a valid solution, which concludes our proof.  
\end{proof}

\section{Discussion}\label{sec:discussion}

In this section, we discuss some high-level aspects of our formulation and results, possible extensions, and limitations.

\paragraph{Scoring Functions} The scoring function can be exploited to capture various notions of similarity between instructions. Namely, it can be more involved than simple equality checking of its arguments~\cite{seqalign,ssafuncmerg}, and the returned score can be proportional to the similarity of the two instructions. More generally, we can define the scoring function to consider a pair of (sequences of) instructions equivalent if they induce the same transfer function over the program state. In this sense, we can let $\delta$ consider, say, the instructions \texttt{x++;y++} and \texttt{y++;x++} as equivalent and thus candidates for alignment as they both induce the same update to the program's state. However, in general, for longer sequences, this will require costly SMT calls to check equivalence.

\paragraph{Balancing Depth and Branching} We remark that any function can be rewritten so that its branching factor is at most $2.$ However, this transformation restricts the set of possible linearizations and moreover increases nesting depth. For instance, the program $[A\br B\br C]$ (with $b=3,d=1$) can be rewritten as $[[A\br B]\br C]$ (with $b=2,d=2$). The former can be linearized as $BCA$ whereas the latter cannot. This loss of information applies for any other way of rewriting. Nevertheless, such transformation provides a good approximation when we aim to reduce the branching factor. We can also do the reverse transformation to reduce nesting depth and increase the branching factor whilst gaining more linearizations. A well-engineered implementation should employ these two transformations in pre-processing to find a balance where both $b$ and $d$ are small, without sacrificing too much precision.

\paragraph{Fine-tuning Pair Selection} Our algorithm's runtime is inherently dependent on $b$ and $d.$ One can use this as a heuristic to only try pairs of functions that have smaller $b, d,$ thus reducing the optimization runtime. We note, however, that there is no direct tradeoff between the estimated similarity of selected pairs and our runtime. The algorithm runs fast on two non-similar functions if $b,d$ are small and will be slow if $b,d$ are large, even if the two functions are highly similar.

\paragraph{Variable Renaming} Another aspect that can increase the effectiveness of function merging is variable renaming, where we can rename variables in the two functions to increase their similarity and thus the potential for merging. This can be integrated into our algorithms as follows: Given a pair of functions, the merging algorithm can be run on that pair multiple times under different variable renamings, and we then choose the best merging over all renaming choices. Alternatively, we can interleave renaming with merging by augmenting the dynamic programming state with extra information to indicate which variables have been renamed to which. We note, however, that there are exponentially many variable renamings and thus such extension would require heuristics to approximate which renamings are profitable with respect to merging.

\paragraph{Encoding Conditionals} Conditionals were not encoded in our formulation, and incorporating them depends on the exact representation of conditionals as well as the procedure that merges two already-aligned functions. Nevertheless, the most likely scenario is that function merging is applied to a low-level intermediate representation language. There, compilers like GCC or LLVM split a condition into an instruction computing the predicate followed by a conditional jump statement, and append an unconditional jump to the end of the condition's body. For instance, a statement \texttt{if(X1)\{Y1\} elif(X2)\{Y2\}; Z} can be transformed into \texttt{b1=..; cond-jmp b1 L1; Y1; jmp L1'; b2=..; cond-jmp b2 L2; Y2; jmp L2'; Z}. This can be abstracted as:
\begin{center}
\texttt{[b1=..$\seq$ cond-jmp b1 L1 $\seq$ Y1 $\seq$ jmp L1'$\br$b2 =..$\seq$ cond-jmp b2 L2 $\seq$ Y2 $\seq$ jmp L2'] $\seq$ Z},  
\end{center}
which soundly accounts for conditionals. Merging this with another function will allow us to match the computation of \texttt{b1}, \texttt{b2} with a corresponding instruction in the other function.

\paragraph{Branch Fusion} Our techniques can be adapted to merge any two programs that use standard sequential and branching constructs, and not just functions. In particular, our results carry over to the closely related problem of \emph{branch fusion}, where the goal is to merge two similar branches of an \texttt{if-else} statement into a single branch in order to reduce code size~\cite{diverg,hybf}.

\paragraph{The Need for Well-nesting} Our approach heavily relies on the well-nesting of branches and only considers reorderings local to each branching construct. We note that our results can be stated more generally in terms of reordering the successors of nodes at the CFG-level, even when the CFG is not well-nested. However, this would crucially depend on the procedure used to identify meaningful points of reordering so as to balance the branching factor and the depth of the reordering.

\paragraph{Handling Code Motion and Data Dependencies} An interesting direction for future work is to extend our approach so that it supports code motion based on an arbitrary set of data dependence rules between the program statements, which do not necessarily follow the well-nested structure of the program. For example, an optimizing compiler might choose to move out a statement from the body of a loop, while keeping the program equivalent to its original version. This would enable new reorderings, which might open up the possibility of better alignments.


\section{Conclusion}\label{sec:conclusion}

In this work, we considered an expressive formulation for function merging that incorporates branch reordering. We then provided a thorough algorithmic and complexity-theoretic analysis of the problem. We showed that the problem is \NP-hard even under severe restrictions on the input programs, and we identified natural parameters that make the problem fixed-parameter tractable.

\section*{Acknowledgments and Notes}

We are thankful to the anonymous PLDI reviewers for their comments which helped improve our work significantly. This work was supported by the ERC Starting Grant 101222524 (SPES).  K. Kochekov, T. Shu and A.K. Zaher were supported by the Hong Kong PhD Fellowship Scheme (HKPFS). Author names are ordered alphabetically. A.K. Zaher is the corresponding author.

\bibliographystyle{ACM-Reference-Format}
\bibliography{refs}

\appendix

\section{Appendix} \label{sec:app}

\subsection{Proof of Lemma~\ref{lem:correctness}} \label{sec:appsec-1}

\lemcorrectness*

\begin{proof}
  We first give a quick refresher on the classical \SA algorithm. Given two strings $S_1, S_2,$ we define dynamic programming states for $i\in\{1,\dots, |S_1|\}, j\in\{1,\dots, |S_2|\}$ as: $$dp[i, j] = \SA(S_1[i,\dots, |S_1|],S_2[j,\dots, |S_2|],\delta).$$ That is, $dp[i, j]$ is the score of the optimal alignment between the suffixes $S_1[i,\dots, |S_1|]$ and $S_2[j,\dots, |S_2|].$ Each state is computed by exhausting all possibilities of the first column in the alignment between those suffixes, and then we handle the rest of the problem recursively. Namely, we have the following transitions for $i\in\{1,\dots, |S_1|\}, j\in\{1,\dots, |S_2|\}$:
\begin{align}\label{eqn:dp}
  dp[i, j] = \max \begin{cases}
                    \delta(S_1[i], S_2[j]) + dp[i+1, j+1], \\
                    \delta(S_1[i], -) + dp[i+1, j],        \\
                    \delta(-, S_2[j]) + dp[i, j+1].
                  \end{cases}
\end{align}
If $i=|S_1|+1\land j=|S_2|+1,$ then $dp[i, j] = 0.$ Otherwise, if $i=|S_1|+1,$ then $dp[i, j] = \delta(-, S_2[j]) + dp[i, j+1],$ and if $j=|S_2|+1,$ then $dp[i, j] = \delta(S_1[i], -) + dp[i+1, j].$

Back to \FMBR, let $\mathbf{S}_i$ be the set of all possible linearizations of $P_i,$ for $i\in\{1, 2\}.$ By definition, we have $\FMBR(P_1, P_2,\delta) = \max_{S_1\in\mathbf{S}_1, S_2\in\mathbf{S}_2} \SA(S_1, S_2,\delta).$ If we show that Algorithm~\ref{alg:fmbr} explores all pairs of linearizations $S_1, S_2,$ and that it moreover computes $\SA(S_1, S_2,\delta),$ ultimately taking the maximum over all pairs, then we are done. 

Consider a sequence of states $$(u_1^0, X_1^0, u_2^0, X_2^0), (u_1^1, X_1^1, u_2^1, X_2^1),\dots, (u_1^m, X_1^m, u_2^m, X_2^m)$$ visited along some series of recursive calls in Algorithm~\ref{alg:fmbr} starting from $(u_1^0, X_1^0, u_2^0, X_2^0):=(l_1, \emptyset, l_2, \emptyset).$ Fix $i\in\{1, 2\}.$ Consider the restricted sequence $(u_i^{j_1}, X_i^{j_1}),\dots, (u_i^{j_{m'}}, X_i^{j_{m'}}), m'\leq m,$ which ignores the $u_{3-i},X_{3-i}$ components and merges consecutive duplicates. This restricted sequence corresponds to a sequence of states visited by Algorithm~\ref{alg:enum-eulerian} when run on $\BG(P_i),$ and vice versa. In particular, reading $\Sigma$-labels of such restricted sequence gives a linearization in $\mathbf{S}_i.$ By Lemma~\ref{lem:enum-eulerian}, the union of all restricted sequences for the runs of Algorithm~\ref{alg:fmbr} is precisely the set $\mathbf{S}_i.$ 

Then, note that Algorithm~\ref{alg:fmbr} interleaves two runs of Algorithm~\ref{alg:enum-eulerian} on $\BG(P_1),\BG(P_2).$ Therefore, for every $S_1\in \mathbf{S}_1$ and $S_2\in\mathbf{S}_2,$ there is a sequence of states visited by Algorithm~\ref{alg:fmbr} whose restricted sequence over the $u_1, X_1$ components spells out $S_1$ and over the $u_2, X_2$ components spells out $S_2.$ If we fix one such pair, then the algorithm's transitions coincide with those of the \SA~algorithm sketched above, which computes $\SA(S_1, S_2,\delta).$ Namely, the base case is when $u_1 = u_2 = \bot,$ in which case we correctly return $0$ (lines 5-6). If only one of the $S_i$'s is fully traversed, then we apply the corresponding transition for the other string (lines 29-32). Otherwise, if both $S_1, S_2$ have more instructions, then we apply the transitions in Equation~(\ref{eqn:dp}) of the \SA algorithm (line 34). Given that we take the maximum over all pairs of linearizations, we get the desired result.
\end{proof}

\subsection{Proof of Theorem~\ref{thm:fpt1}} \label{sec:appsec-2}

\thmfptone*

\begin{proof}
  For correctness, by Lemma~\ref{lem:correctness}, we only need to show correctness of the updated {\scshape Upd-$\ast$} and {\scshape Contains} functions in Algorithm~\ref{alg:signature-ops}. It is not hard to see that with these definitions, invoking {\scshape Contains}$(u_i, v,X_i)$ in the modified Algorithm~\ref{alg:fmbr}, where $X_i$ is a signature, returns the same boolean value as invoking $\text{\scshape Contains}(u_i, v, X_i)$ in the original Algorithm~\ref{alg:fmbr}, where $X_i$ is a set of visited nodes. This implies that the modified Algorithm~\ref{alg:fmbr} explores the state space in the same way as the original Algorithm~\ref{alg:fmbr}, and thus computes the same value. 
  
  Namely, we start with a signature $X_i$ containing one empty record, corresponding to the empty set initially passed in the original algorithm. At {\scshape Upd-en}, we increase the depth and go to a previously unvisited node, so we push a new empty record on the stack. At {\scshape Upd-seq}, we move to the neighbouring same-depth node, which cannot have been visited before, so we get rid of the current top record and replace it with a new empty record. At {\scshape Upd-ex}, we decrease the depth and pop the current record from the stack. Moreover, we need to update the record of the (immediate) enclosing branching construct to inform it that $u_i$'s branch has been fully visited, which we do by inserting the index of that branch. At {\scshape Upd-done}, we have $u_i=r_i$ and the signature has exactly one record, so we pop it to get the empty signature $\emptysgn,$ which will persist until $P_i$ is fully processed, because $u_i$ will be set to $\bot.$ Finally, {\scshape Contains}$(u_i, v, X_i)$ checks whether the branch of $u_i$ identified by $v$ is not yet processed, which is a simple membership check on the top record of the stack.

  For the time complexity, note a signature for $P_i$ can be directly represented as a stack, with each of its elements being a subset of $\{1,\dots, b_i\}.$ Clearly, operations on a given stack element, e.g., inserting an index or checking membership, can be done in $\bigO(b_i)$ time. Manipulating the stack itself, e.g., copying, pushing, popping, can be naively done in $\bigO(b_i h)$ time, where $h$ is the height of the stack, which can be as large as $\bigO(d_i).$ Thus, each call to {\scshape Upd}-$\ast$ and {\scshape Contains} takes $\bigO(b_id_i)$ time.
  
  As argued in Lemma~\ref{lem:state-space}, we have $\bigO(2^{b_i\cdot d_i})$ possible signatures for $P_i.$ Since we remember an index and a signature for each program $P_i,$ we get $\bigO(2^{b_1 d_1 + b_2 d_2}\cdot n_1n_2)$ dynamic programming states. In each state, for $i\in\{1, 2\},$ we make $\bigO(b_i)$ calls to {\scshape Contains} and $\bigO(1)$ calls to {\scshape Upd}-$\ast,$ each taking $\bigO(b_i d_i)$ time, which is clearly upper bounded by $\bigO((b_1+b_2)^2(d_1+d_2)).$ Other routine operations, such as checking and fetching values from the memoization table, can also be done within the same time bound. This gives the claimed time complexity.
\end{proof}

\subsection{Proof of Lemma~\ref{lem:pds-reach}} \label{sec:appsec-3}

\lempdsreach*

\begin{proof}  
  Let $c_1:=\langle (l_1, [\emptyset], l_2, [\emptyset]), \$ \rangle, c_2:=\langle ( \bot,\emptysgn,\bot,\emptysgn), \$ \rangle.$ Define $H$ as the subgraph of the configuration graph $G_{\pds_1}$ induced by nodes reachable from the initial configuration $c_1.$ By construction of $\pds_1,$ we have that $H$ is isomorphic to the graph of recursive calls of Algorithm~\ref{alg:fmbr} when run on $(P_1, P_2,\delta).$ Note that this implies that $H$ is finite and acyclic and thus $\paths(c_1, c_2)$ is finite, making the value $\reach(c_1, c_2)$ well-defined. By definition of how {\scshape FuncMerg} computes its values, we have that {\scshape FuncMerg}$(l_1,[\emptyset],l_2,[\emptyset])$ is the $\add$ of path weights over paths $\rho$ in the recursion graph s.t. $\rho$ starts from $(l_1,[\emptyset],l_2,[\emptyset])$ and ends at $(\bot,\emptysgn,\bot,\emptysgn).$ By the aforementioned isomorphism, we have $\reach(c_1, c_2) = \text{\scshape FuncMerg}(l_1,[\emptyset],l_2,[\emptyset]).$ By Theorem~\ref{thm:fpt1}, we have $\FMBR(P_1, P_2,\delta)=\text{\scshape FuncMerg}(l_1,[\emptyset],l_2,[\emptyset])$ and thus $\reach(c_1, c_2) = \FMBR(P_1, P_2,\delta),$ as desired.
\end{proof}

\subsection{Proof of Theorem~\ref{thm:fpt2}} \label{sec:appsec-4}

Before proving Theorem~\ref{thm:fpt2}, we make some adjustments to allow the use of the reachability algorithm of~\cite{wpds} which, given configurations $c, c',$ computes $\reach(c, c').$ The algorithm's running time is polynomial in the size of input WPDS and the \emph{length of the longest ascending chain} in the semiring. Since our semiring can have infinite ascending chains, this is not directly applicable. However, we will show an alternative more fine-grained condition that suffices for the algorithm to work efficiently, and we will show that this condition is satisfied by the WPDS $\pds_2$ constructed for our problem.

Specifically, given WPDS $\pds=(\pdsstates, \pdsstk, \pdstrans,\wt)$ over semiring $\semiring,$ and configurations $c,c',$ \cite[Algorithm 1]{wpds} computes $\reach(c, c')$ by means of a saturation procedure that iteratively adds transitions to an $\pds$-automaton (which we do not formally define here) representing the set of configurations that can reach $c'.$ The algorithm runs in $\bigO(|\pdsstates| \cdot |\pdstrans|^2 \cdot X)$ operations, where $X$ is the number of times a transition $t$ of the $\pds$-automaton is updated to a new, previously unencountered, semiring value~\cite[pg. 219]{wpds}. Here, we observe that at any moment, the value of $t$ can be expressed as $\bigoplus_{\rho\in Q} \wt(\rho),$ where $Q$ is a set of paths in the configuration graph $G_{\pds}$ of $\pds.$ That is, the values encountered while running the algorithm are derived from the input WPDS and are not necessarily arbitrary elements of the semiring. Define the \emph{diversity} of $\pds,$ denoted $div(\pds),$ as the smallest size of a set $Y$ s.t. for any set of paths $Q$ in $G_{\pds},$ we have $\bigoplus_{\rho\in Q} \wt(\rho) \in Y.$ It follows that $X$ is upper bounded by the diversity of $\pds,$ giving us the following alternative running time bound: 

\begin{theorem}[\cite{wpds}]\label{thm:wpdsreach}
   Given a WPDS $\pds$ and configurations $c, c',$
  $\reach(c, c')$ can be computed in $\bigO(|\pdsstates| \cdot |\pdstrans|^2 \cdot div(\pds))$ basic operations.
\end{theorem}
Here, basic operations refer to semiring operations, comparisons, and manipulations of stack elements and states. We next show that the diversity of $\pds_2$ defined in Section~\ref{sec:algorithms-2} is polynomial:

\begin{proposition}\label{prop:div}
  $div(\pds_2)\in\bigO(n_1+n_2).$
\end{proposition}
\begin{proof}
Take a path $\rho$ in $G_{\pds_2}.$ If one of $\rho$'s edges has weight $-\infty,$ then $\wt(\rho)=-\infty.$ Otherwise, observe that:
\begin{myitemize}
  \item $\rho$ has at most $2(n_1+n_2)$ edges $e$ with $\wt(e)\neq 0.$\footnote{Recall that all the transitions of the form $((u_1, u_2, X_2)_j, S)\pdsmove((u_1, u_2, X_2), S\cup\{j\})$ have weight $0.$} Indeed, at every such edge of $\rho,$ we traverse a new edge of the branching graph $\BG(P_1)$ or $\BG(P_2),$ and $\BG(P_i)$ has at most $2|P_i|=2n_i$ edges.
  \item The weight of every edge in $\rho$ is in $\{L,\dots, R\}.$
\end{myitemize}
These two facts imply that the weight of $\rho$ is in $\{2L(n_1+n_2),\dots, 2R(n_1+n_2)\}.$ Recall that $L\leq 0\land 1\leq R.$ Overall, we have $\wt(\rho)\in\{2L(n_1+n_2),\dots, 2R(n_1+n_2)\}\cup\{-\infty\}.$ Since $\add = \max,$ for any set $Q$ of paths in $G_{\pds_2},$ we have that $\bigoplus_{\rho\in Q} \wt(\rho)$ is also in $\{2L(n_1+n_2),\dots, 2R(n_1+n_2)\}\cup\{-\infty\},$ which has size $\bigO(n_1+n_2).$
\end{proof}

We are now ready to prove Theorem~\ref{thm:fpt2}.

\thmfpttwo*

\begin{proof}
  The algorithm is simple: We first construct the WPDS $\pds_2$ as described, and then use Theorem~\ref{thm:wpdsreach} to compute $\nu := \reach(\langle (l_1, l_2, [\emptyset]), \emptyset\$\rangle,\langle (\bot,\bot,[]),\$\rangle),$ and then return $\nu.$ 
  
  \paragraph{Correctness} We rely on the proof of Lemma~\ref{lem:pds-reach}. In the stack-less WPDS $\pds_1,$ let $H$ be the subgraph of the configuration graph $G_{\pds_1}$ induced by nodes reachable from $\langle (l_1, [\emptyset], l_2, [\emptyset]), \$\rangle.$ Similarly, define $K$ as the subgraph of $G_{\pds_2}$ induced by nodes reachable from $\langle (l_1, l_2, [\emptyset]), \emptyset\$\rangle.$ Further, let $K'$ be the result of taking all edges in $K$ whose respective transition is of the form $((u_1, u_2, X_2)_j, S)\pdsmove((u_1, u_2, X_2), S\cup\{j\}),$ then merging their endpoints. By construction of $\pds_2$ and definition of configuration graphs, it is clear that $K'$ is isomorphic to $H,$ and thus $\reach_{K'}(\langle (l_1, l_2, [\emptyset]), \emptyset\$\rangle,\langle (\bot,\bot,[]),\$\rangle)=\reach_H(\langle (l_1, [\emptyset], l_2, [\emptyset]), \$ \rangle,\langle(\bot, \emptysgn, \bot, \emptysgn),\$\rangle)=\FMBR(P_1, P_2,\delta),$ by Lemma~\ref{lem:pds-reach}. By unmerging the edges of $K'$ back to get $K,$ all paths only get an extra edge of weight $0,$ and thus reachability values are preserved, implying: $$\nu=\reach_K(\langle (l_1, l_2, [\emptyset]), \emptyset\$\rangle,\langle (\bot,\bot,[]),\$\rangle)=\FMBR(P_1, P_2,\delta).$$
  
  \paragraph{Running time} $\pds_2$ has $|\pdsstates_2| =2^{\bigO(b_2d_2)}b_1 n_1n_2$ states and $|\Delta_2|=2^{\bigO(b_2d_2+b_1)} n_1n_2$ transitions. Note that any polynomial term $\param^{\bigO(1)}$ can be absorbed into a $2^{\bigO(\param)}$ term. The construction of the WPDS can be done in time $|\Delta_2|\cdot (b_1b_2d_2)^{\bigO(1)}=2^{\bigO(b_2d_2+b_1)} n_1n_2$ time. By Theorem~\ref{thm:wpdsreach}, computing $\nu$ takes $\bigO(|\pdsstates_2|  |\Delta_2|^2  div(\pds_2))$ basic operations. Here, a basic operation potentially involves manipulating a signature for $u_2$ or manipulating a \emph{single record} of $u_1$'s signature, which is bounded by $(b_1b_2d_2)^{\bigO(1)}.$ Note that at no point do we incur a dependence on $d_1.$ By substitution and Proposition~\ref{prop:div}, we get the claimed running time of $2^{\bigO(b_2d_2+b_1)}  (n_1n_2)^3(n_1+n_2).$
\end{proof}

\subsection{Proof of Corollary~\ref{cor:notxp}}\label{sec:appsec-5}

\cornotxp*

\begin{proof}
  Suppose for the sake of a contradiction that \FMBR~parameterized by $d_i, b_{3-i}, d_{3-i}$ is in \XP. This implies that there is an algorithm for \FMBR~running in time $\bigO(n^{f(d_{i}, b_{3-i}, d_{3-i})})$ for some computable function $f.$ We can use this algorithm to solve \FMBR~restricted to inputs where $d_i=1$ and $b_{3-i} = d_{3-i} = 0.$ The algorithm will run in time $\bigO(n^{f(1, 0, 0)}).$ Since $f(1, 0, 0)$ is an absolute constant independent of $n$ and the parameters, this is a polynomial-time algorithm for the restricted version of \FMBR. However, by Theorem~\ref{thm:dbdhard}, this restricted version is \NP-hard, leading to a contradiction, unless $\PTIME= \NP.$
\end{proof}

\end{document}